\documentclass[11pt]{article}

\usepackage[dvips]{graphicx}
\usepackage{amsmath,amsthm,amscd}
\usepackage{amssymb}
\usepackage{fancybox}
\usepackage{latexsym}
\usepackage[dvips]{color}

\def\KL{\mathrm{KL}}

\def\Real{\Rbb}
\def\Rbb{\mathbb{R}}

\def\<{\langle}
\def\>{\rangle}

\def\st{\text{\rm s.\,t.\ }}
\def\subto{\text{\rm subject to }}

\newcommand{\tr}[1]{{\mathrm{tr}(#1)}}

\newtheorem{theorem}{Theorem}
\newtheorem{lemma}[theorem]{Lemma}

\newtheorem{definition}{Definition}
\newtheorem{example}{Example}
\newtheorem{remark}{Remark}

\def\0{\mbox{\bf 0}}
\def\1{\mbox{\bf 1}}
\def\2{\mbox{\bf 2}}
\def\3{\mbox{\bf 3}}
\def\4{\mbox{\bf 4}}
\def\5{\mbox{\bf 5}}
\def\6{\mbox{\bf 6}}
\def\7{\mbox{\bf 7}}
\def\8{\mbox{\bf 8}}
\def\9{\mbox{\bf 9}}

\definecolor{cyan}{cmyk}{1,0,0,0}
\definecolor{lightcyan}{cmyk}{0.5,0,0,0}
\definecolor{pastelcyan}{cmyk}{0.25,0,0,0}
\definecolor{magenta}{cmyk}{0,1,0,0}
\definecolor{yellow}{cmyk}{0,0,1,0}
\definecolor{lightyellow}{cmyk}{0,0,0.5,0}
\definecolor{pastelyellow}{cmyk}{0,0,0.25,0}
\definecolor{black}{cmyk}{0,0,0,1}
\definecolor{darkgray}{cmyk}{0,0,0,0.75}
\definecolor{gray}{cmyk}{0,0,0,0.5}
\definecolor{lightgray}{cmyk}{0,0,0,0.25}
\definecolor{white}{cmyk}{0,0,0,0}
\definecolor{red}{cmyk}{0,1,1,0}
\definecolor{orange}{cmyk}{0,0.5,1,0}
\definecolor{scarlet}{cmyk}{0,1,0.5,0}
\definecolor{brown}{cmyk}{0.5,0.75,1,0}
\definecolor{camel}{cmyk}{0.25,0.375,0.5,0}
\definecolor{cream}{cmyk}{0,0.2,0.3,0}
\definecolor{green}{cmyk}{1,0,1,0}
\definecolor{lightgreen}{cmyk}{0.5,0,0.5,0}
\definecolor{pastelgreen}{cmyk}{0.25,0,0.25,0}
\definecolor{mossgreen}{cmyk}{0.64,0.4,1,0}
\definecolor{yellowgreen}{cmyk}{0.5,0,1,0}
\definecolor{skyblue}{cmyk}{0.4,0.16,0,0}
\definecolor{royal}{cmyk}{1.0,0.5,0,0}
\definecolor{navyblue}{cmyk}{0.9,0.75,0.5,0}
\definecolor{lightnavy}{cmyk}{0.4,0.3,0.2,0}
\definecolor{blue}{cmyk}{1,1,0,0}
\definecolor{lightblue}{cmyk}{0.5,0.5,0,0}
\definecolor{lavender}{cmyk}{0.25,0.25,0,0}
\definecolor{violet}{cmyk}{0.75,1,0.25,0}
\definecolor{purple}{cmyk}{0.5,1,0.5,0}
\definecolor{pink}{cmyk}{0,0.5,0,0}
\definecolor{pastelpink}{cmyk}{0,0.25,0,0}

\def\PD{\mathrm{PD}}

\title{A Bregman Extension of quasi-Newton updates I:\\ 
An Information Geometrical framework}
\author{
  Takafumi Kanamori\\ Nagoya University \\ \tt{kanamori@is.nagoya-u.ac.jp}
  \and
  Atsumi Ohara\\ Osaka University\\ \tt{ohara@sys.es.osaka-u.ac.jp}
 }

\date{}

\begin{document}
\maketitle
 
\begin{abstract}
 We study quasi-Newton methods from the viewpoint of information geometry induced
 associated with Bregman divergences. 
 Fletcher has studied a variational problem which derives the approximate Hessian update 
 formula of the quasi-Newton methods. 
 We point out that the variational problem is identical to optimization of the
 Kullback-Leibler divergence, which is a discrepancy measure 
 between two probability distributions. 
 The Kullback-Leibler divergence for the multinomial normal distribution corresponds to
 the objective function Fletcher has considered. 
 We introduce the Bregman divergence as an extension of the Kullback-Leibler divergence,
 and derive extended quasi-Newton update formulae based on the variational problem with
 the Bregman divergence. 
 As well as the Kullback-Leibler divergence, the Bregman divergence introduces 
 the information geometrical structure on the set of positive definite matrices. 
 From the geometrical viewpoint, we study 
 the approximation Hessian update, the invariance property of the update formulae, 
 and the sparse quasi-Newton methods. 
 Especially, we point out that the sparse quasi-Newton method is closely related 
 to statistical methods such as the EM-algorithm  and the boosting algorithm. 
 Information geometry is useful tool not only to better understand the quasi-Newton
 methods but also to design new update formulae. 
\end{abstract}

\section{Introduction}
\label{sec:Introduction}
The main purpose of this article is to study the quasi-Newton methods from the view point
of dualistic geometry or in other word {\em information geometry} 
\cite{AmariNagaoka00,ohara05:_geomet_posit_defin_matric_and,
murata04:_infor_geomet_u_boost_bregm_diver}. 
Let us consider the unconstrained optimization problem 
\begin{align}
 \label{eqn:main_opt_problem}
 \text{minimize}\  f(x),\quad x\in\Real^n, 
\end{align}
in which the function $f:\Real^n\rightarrow\Real$ is twice continuously
differentiable on $\Real^n$. 
The quasi-Newton method is known to be one of the most successful methods for
unconstrained function optimization. 
In quasi-Newton method a sequence $\{x_k\}_{k=0}^{\infty}\subset\Real^n$ is successively
generated in a manner such that $x_{k+1}=x_k-\alpha_k B_k^{-1}\nabla f(x_k)$, where
$\alpha_k$ is a step length computed by a line search technique. The matrix $B_k$ is a
positive definite matrix which is expected to approximate the Hessian matrix 
$\nabla^2 f(x_k)$. The matrix $B_k$ and the step length $\alpha_k$ are designed such that
the sequence $x_k$ converges to a local minima of the problem
\eqref{eqn:main_opt_problem}. 
For the step length, the Wolfe condition \cite[Section 3.1]{nocedal99:_numer_optim} is a
standard criterion to determine the value of $\alpha_k$. 
In terms of the approximate Hessian matrix, mainly there are two methods of updating $B_k$
to $B_{k+1}$; one is called the DFP formula and the other is called the BFGS formula. 

We introduce the DFP and the BFGS methods. 
Let $s_k$ and $y_k$ be column vectors defined by
\begin{align*}
 s_k=x_{k+1}-x_k=-\alpha_kB_k^{-1}\nabla f(x_k),
 \qquad y_k=\nabla f(x_{k+1})-\nabla f(x_k), 
\end{align*}
and suppose that $s_k^\top y_k>0$ holds. 
In the DFP formula the approximate Hessian matrix $B_k$ is updated such that 
\begin{align}
 \label{eqn:DFP-update-formula}
 B_{k+1}
 &=
 B^{DFP}[B_k;s_k,y_k] 
 := B_k-\frac{B_ks_ky_k^\top+y_ks_k^\top B_k}{s_k^\top y_k}
 +s_k^\top B_ks_k \frac{y_k y_k^\top}{(s_k^\top y_k)^2}+\frac{y_ky_k^\top}{s_k^\top y_k}. 
\end{align}
In the BFGS update formula, the matrix $B_{k+1}$ is defined by 
\begin{align}
 \label{eqn:BFGS-update-formula}
 B_{k+1}
 &= B^{BFGS}[B_k;s_k,y_k]
 :=B_k-\frac{B_ks_ks_k^\top B_k}{s_k^\top B_ks_k}+\frac{y_ky_k^\top}{s_k^\top y_k}, 
\end{align}
Under the condition that $B_k\in\mathrm{PD}(n)$ and $s_k^\top y_k>0$, 
the matrices $B^{DFP}[B_k;s_k,y_k]$ and $B^{BFGS}[B_k;s_k,y_k]$ are also positive definite
matrices. 
If there is no confusion, the update formulae $B^{DFP}[B;s,y]$ and $B^{BFGS}[B;s,y]$ are
written as $B^{DFP}[B]$ and $B^{BFGS}[B]$, respectively. 
In practice, the Cholesky decomposition of $B_k$ is successively updated in order to
compute the search direction $-B_k^{-1}\nabla f(x_k)$ efficiently
\cite{gill72:_quasi_newton_method_for_uncon_optim}. 
Note that the equality 
\begin{align*}
 B^{DFP}[B;s,y]^{-1}=B^{BFGS}[B^{-1};y,s]
\end{align*}
holds. Hence, we can derive the update formulae for the inverse $H_k=B_k^{-1}$ without
inversion of matrix. 

Both the DFP and the BFGS methods are derived from variational problems over the set of
positive definite matrices \cite{fletcher91:_new_resul_for_quasi_newton_formul}.  
Let $\PD(n)$ be the set of all $n$ by $n$ symmetric positive definite matrices, and 
the function $\psi:\mathrm{PD}(n)\rightarrow\Real$ be a strictly convex function 
over $\PD(n)$ defined by 
\begin{align*}
 \psi(A)=\tr{A}-\log\det{A}. 
\end{align*}
Fletcher \cite{fletcher91:_new_resul_for_quasi_newton_formul} has shown that the DFP
update formula \eqref{eqn:DFP-update-formula} is obtained as the unique solution 
of the constraint optimization problem, 
\begin{align*}
 \min_{B\in\mathrm{PD}(n)}\ \psi(B_k^{1/2}B^{-1}B_k^{1/2})\quad \subto\ Bs_k=y_k, 
\end{align*}
where $A^{1/2}$ for $A\in\PD(n)$ is the matrix satisfying $A^{1/2}\in\PD(n)$
and $(A^{1/2})^2=A$. 
The BFGS formula is also obtained as the optimal solution of 
\begin{align*}
 \min_{B\in\mathrm{PD}(n)}\ \psi(B_k^{-1/2}BB_k^{-1/2})\quad \subto\ Bs_k=y_k, 
\end{align*}
in which $B_k^{-1/2}$ denotes $(B_k^{-1})^{1/2}$ or equivalently $(B_k^{1/2})^{-1}$. 

It will be worthwhile to point out that the function $\psi$ is identical to
Kullback-Leibler(KL) divergence \cite{AmariNagaoka00,kullback51:_infor_and_suffic} 
up to an additive constant. 
For $P, Q\in\PD(n)$, the KL-divergence is defined by
\begin{align*}
 \mathrm{KL}(P,Q)
 &=\tr{PQ^{-1}}-\log\det(PQ^{-1})-n
\end{align*}
which is equal to $\psi(Q^{-1/2}PQ^{-1/2})-n$. 
The KL-divergence is regarded as a generalization of squared distance. 
Using the KL-divergence, we can represent the update formulae as the optimal solutions of
the following minimization problems, 
\begin{align}
 \text{(DFP)}\qquad&\min_{B\in\mathrm{PD}(n)}\  \mathrm{KL}(B_k,B)\quad \subto\  Bs_k=y_k,
 \label{eqn:DFP}\\
 \text{(BFGS)}\qquad&\min_{B\in\mathrm{PD}(n)}\ \mathrm{KL}(B,B_k)\quad \subto\  Bs_k=y_k. 
 \label{eqn:BFGS}
\end{align} 
The KL-divergence is asymmetric, that is, $\KL(P,Q)\neq \KL(Q,P)$ in general. 
Hence the above problems will provide different solutions. 

In the information geometry \cite{AmariNagaoka00}, 
the KL-divergence defines a geometrical structure over the space
of probability densities. Statistical inference such that the maximum likelihood estimator
is better understood based on the geometrical intuition. 
Originally, the KL-divergence is defined as the discrepancy measure between two multinomial
normal distributions with mean zero. 
In this paper, we show that the information geometrical approach is useful to understand
the behaviour of quasi-Newton methods. 
On the set of positive definite matrices, $\PD(n)$, we define the so-called Bregman
divergence which is an extension of the KL-divergence. 
The Bregman divergence induces a dualistic geometrical structure on $\PD(n)$. Then we can
derive new Hessian update formulae based on the Bregman divergence. 
We present a geometrical view of quasi-Newton updates, and discuss the relation between
the Hessian update formula and the statistical inference based on the information
geometry. 


Here is the brief outline of the article. 
In Section \ref{sec:Elements_Information_Geometry},
we introduce the elements of information geometry based on the Bregman divergence, 
especially over the set of positive definite matrices. 
In Section \ref{sec:potential-function_quasi-Newton}, 
an extended quasi-Newton formula is derived from the Bregman divergence. 
Section \ref{sec:Invariance} is devoted to discuss the invariance property of the
quasi-Newton update formula under the group action. 
In Section \ref{sec:Sparse-V-quasi-Newton}, we discuss 
the sparse quasi-Newton methods \cite{yamashita08:_spars_quasi_newton_updat_with} from the
viewpoint of the information geometry, and point out that the sparse quasi-Newton method
is closely related to statistical methods such as the EM-algorithm 
\cite{mclachlan08:_em_algor_and_exteny} or the boosting algorithm
\cite{FreundSchapire97,murata04:_infor_geomet_u_boost_bregm_diver}. 
We conclude with a discussion and outlook in Section \ref{sec:Concluding_Remarks}. 
Some proofs of the theorems are postponed to Appendix. 


Throughout the paper, we use the following notations: 
The set of positive real numbers are denoted as $\Real_+\subset\Real$. 
Let $\det{A}$ be the determinant of square matrix $A$, and 
$\mathrm{GL}(n)$ denotes the set of $n$ by $n$ non-degenerate real matrices. 
$\mathrm{SL}(n)\subset \mathrm{GL}(n)$ is the set of $n$ by $n$ non-degenerate real
matrices with determinant $1$, that is, 
$\mathrm{SL}(n)=\{A\in\mathrm{GL}(n)~|~\det{A}=1\}$. 
The set of all $n$ by $n$ real symmetric matrices is denoted as $\mathrm{Sym}(n)$, and 
let $\mathrm{PD}(n)\subset\mathrm{GL}(n)\cap\mathrm{Sym}(n)$ be the set of $n$ by $n$
symmetric positive definite matrices. 
For $P\in\mathrm{PD}(n)$, the square root of $P$ is denoted as $P^{1/2}$ which is defined
as $P$
For a vector $x$, $\|x\|$ denotes the Euclidean norm. 
For two square matrices $A,\,B$, the inner product 
$\<A,B\>$ is defined by $\tr{A B^\top}$, and $\|A\|_F$ is the Frobenius norm defined by
the square root of $\<A,A\>$. 
Throughout the paper we only deal with the inner product of symmetric matrices, and the
transposition in the trace can be dropped.

\section{Bregman Divergences and Dualistic Geometry of Positive Definite Matrices} 
\label{sec:Elements_Information_Geometry}
We introduce Bregman divergences which are regarded as an extension of the
KL-divergence. Then we illustrate a differential geometrical structure defined from the
Bregman divergence over the set of positive definite matrices. 
In sequel sections, we will provide a geometrical interpretation of quasi-Newton methods. 
For general Bregman divergences, however, the quasi-Newton update formula cannot be
obtained in the explicit form. In order to obtain computationally tractable update
formulae, we often use a specific Bregman divergence which is called the $V$-Bregman 
divergence in this article. First, we define general Bregman divergences, and then we
introduce the $V$-Bregman divergence as a special case of general Bregman divergences. 
We will show the associated geometrical structure on the set of positive definite
matrices. 

\subsection{Bregman divergences}
\label{subsec:BregmanDiv}
The Bregman divergence \cite{bregman67:_relax_method_of_findin_common} is defined through
the so-called potential function. Below, we define the Bregman divergence over the set of
positive definite matrices. 
\begin{definition}[Potential function and Bregman divergence]
Let $\varphi:\PD(n)\rightarrow\Real$ be a continuously differentiable, strictly convex
 function that maps positive definite matrices to real numbers.
The function $\varphi$ is referred to as potential function or potential for short. 
Given a potential $\varphi$, the Bregman divergence $D_\varphi(P,Q)$ is defined as 
\begin{align}
 D_\varphi(P,Q)=\varphi(P)-\varphi(Q)-\<\nabla\varphi(Q),P-Q\>
 \label{eqn:def_bregman_div}
\end{align}
 for $P,Q\in\PD(n)$, where $\nabla\varphi(Q)$ is the $n$ by $n$ matrix whose $(i,j)$
 element is given as $\frac{\partial\varphi}{\partial Q_{ij}}(Q)$. 
\end{definition}
The Bregman divergence $D_{\varphi}(P,Q)$ is non-negative and equals zero if and only if
$P=Q$ holds. Indeed, due to the strict convexity of $\varphi$, the function $\varphi(P)$
lies above its tangents $\varphi(Q)+\<\nabla\varphi(Q),P-Q\>$ at $Q$. Hence, the
non-negativity of the Bregman divergence $D_{\varphi}(P,Q)$ is guaranteed. 
Note that $D_{\varphi}(P,Q)$ is convex in $P$ but not necessarily
convex in $Q$. 
Bregman divergences have been well studied in the fields of statistics and machine learning 
\cite{banerjee05:_clust_with_bregm_diver,
dhillon07:_matrix_nearn_probl_with_bregm_diver,
murata04:_infor_geomet_u_boost_bregm_diver}. 

\begin{example}
 \label{example:KL-div}
 For $P\in\PD(n)$ let the function $\varphi$ be $\varphi(P)=-\log\det(P)$. 
 Note that $\varphi(P)$ is a strictly convex function. Then, we have 
 \begin{align*}
  (\nabla\varphi(Q))_{ij}=-\frac{\partial}{\partial Q_{ij}}\log\det Q=-(Q^{-1})_{ji}. 
 \end{align*}
 Hence the corresponding Bregman divergence is 
 \begin{align*}
  D_{\varphi}(P,Q)
  =-\log\det{P}+\log\det{Q}+\<Q^{-1},P-Q\>
  =\<P,Q^{-1}\>-\log\det(PQ^{-1})-n, 
 \end{align*}
 is identical to the KL-divergence on the multivariate normal distribution with mean
 zero \cite{AmariNagaoka00,a.96:_dualis_differ_geomet_of_posit}. 
\end{example}
By replacing the KL-divergence in \eqref{eqn:DFP} or \eqref{eqn:BFGS} with a Bregman
divergence, we will obtain another variational problem for the quasi-Newton method. 
In general, however, update formula cannot be explicitly obtained. 
Below we define a class of Bregman divergences called $V$-Bregman divergence. 
In Section \ref{sec:potential-function_quasi-Newton}, we show that the $V$-Bregman
divergence provides an explicit update formula of the quasi-Newton method. 

We prepare some ingredients to define the $V$-Bregman divergence. 
Let $V:\Real_{+}\rightarrow\Real$ be a strictly convex, decreasing, and third order 
continuously differentiable function. For the derivative $V'$, the inequality $V'<0$ holds
from the condition. 
Indeed, the condition leads to $V'\leq 0$ and $V''\geq0$, and if $V'(z_0)=0$ holds for 
some $z_0\in\Real_+$, then $V'(z)=0$ holds for 
all $z\geq z_0$. 
Hence $V(z)$ is affine function for $z\geq z_0$. This contradicts the strict convexity of
$V$. We define the functions $\nu_V:\Real_+\rightarrow\Real$ and
 $\beta_V:\Real_+\rightarrow\Real$ such that  
 \begin{align*}
  \nu_V(z)=-z V'(z),\qquad 
  \beta_V(z)=\frac{z\nu_V'(z)}{\nu_V(z)}=z\cdot\frac{d}{dz}\log \nu_V(z)
 \end{align*}
 Since $\nu_V(z)>0$ holds for $z>0$, the function $\beta_V$ is well defined on $\Real_+$. 
 The subscript $V$ of $\nu_V$ and $\beta_V$ will be dropped if there is no confusion. 
 We now are ready to present the definition of $V$-Bregman divergence over
 $\mathrm{PD}(n)$. 
\begin{definition}[$V$-Bregman divergence]
 \label{def:V-Bregman-div_potential}
 Let $V:\Real_{+}\rightarrow\Real$ be a function which is strictly convex, decreasing, and
 third order continuously differentiable. Suppose that the functions $\nu$ and $\beta$
 defined from $V$ satisfy the following conditions: 
 \begin{align}
  \label{eqn:beta-condition}
  &\beta(z)<\frac{1}{n}\qquad (z>0)
 \end{align}
 and
 \begin{align}
  \label{eqn:nu-limit-condition}
  &\lim_{z\rightarrow+0}\frac{z}{\nu(z)^{n-1}}=0. 
 \end{align}
 The Bregman divergence defined from the potential $\varphi(P)=V(\det{P})$
 is called $V$-Bregman divergence, and denoted as $D_V(P,Q)$. 
 Not only $V(\det{P})$ but also $V(z)$ is also referred to as potential. 
\end{definition}
As shown in \cite{ohara05:_geomet_posit_defin_matric_and}, the function $V(\det{P})$ is 
strictly convex in $P\in\mathrm{PD}(n)$ if and only if the potential $V$ satisfies 
\eqref{eqn:beta-condition}. 
The $V$-Bregman divergence has the form of 
\begin{align}
 D_V(P,Q)=V(\det{P})-V(\det{Q})+\nu(\det{Q})\<Q^{-1},P\>-n\nu(\det{Q}). 
 \label{eqn:V-Bregman}
\end{align}
Indeed, substituting 
 \begin{align*}
  (\nabla\varphi(Q))_{ij} 
  = \frac{\partial V(\det{Q})}{\partial Q_{ij}}
  = V'(\det{Q})\frac{\partial\det{Q}}{\partial Q_{ij}}
  =-\nu(\det{Q})(Q^{-1})_{ij}, 
 \end{align*}
into \eqref{eqn:def_bregman_div}, we obtain the expression of $D_V(P,Q)$. 
The KL-divergence $\KL(P,Q)$ is represented as $D_V(P,Q)$ with the potential 
$V(z)=-\log z$. Below we show some examples of $V$-Bregman divergence. 
\begin{example}
 \label{example:power-div}
 For the power potential $V(z)=(1-z^\gamma)/\gamma$ with $\gamma<1/n$, 
 we have $\nu(z)=z^\gamma$ and $\beta(z)=\gamma$. Then, we obtain
 \begin{align*}
  D_V(P,Q)=
  (\det{Q})^\gamma
  \bigg\{\<P,Q^{-1}\>+\frac{1-(\det{PQ^{-1}})^\gamma}{\gamma}-n\bigg\}. 
 \end{align*}
 The KL-divergence is recovered by taking the limit of $\gamma\rightarrow0$. 
\end{example}

\begin{example}
 \label{example:bounded-div}
 For $0\leq c<1$, let us define $V(z)=c\log(c z+1)-\log(z)$.
 Then $V(z)$ is a strictly convex and decreasing function, and we obtain 
 \begin{align*}
  \nu(z)=1-c+\frac{c}{cz+1}>0,\qquad 
  \beta(z)=\frac{-c^2z}{(cz+1)(c(1-c)z+1)}\leq 0
 \end{align*}
 for $z>0$. 
 The negative-log potential, $V(z)=-\log z$, is recovered by setting $c=0$.
 The potential satisfies the bounding condition $0<1-c\leq\nu(z)\leq 1$. 
 As shown in the sequel \cite{kanamori10:_bregm_exten_of_quasi_newton_updat_ii},  
 the bounding condition of $\nu$ will be assumed to prove the convergence property of the 
 quasi-Newton method. 
\end{example}

\subsection{Dualistic Geometry defined from Bregman Divergences}
\label{subsec:GeometricalStructure}
The space of positive definite matrices has rich geometrical and algebraic structures
\cite{ohara05:_geomet_posit_defin_matric_and}
Here we introduce dualistic geometrical structure on $\mathrm{PD}(n)$ induced form the
Bregman divergence. 
See \cite{murata04:_infor_geomet_u_boost_bregm_diver,ohara99:_infor_geomet_analy_of_inter}
for details. 

We introduce two coordinate systems on $\mathrm{PD}(n)$. The $\eta$-coordinate system
$\eta:\PD(n)\rightarrow\PD(n)$ is defined as 
\[
\eta(P)=P, 
\]
which is the identity function on $\PD(n)$. 
The definition of the other coordinate system requires the potential $\varphi$ for the
Bregman divergence $D_{\varphi}(P,Q)$ in \eqref{eqn:def_bregman_div}. 
Let us define the $\theta_\varphi$-coordinate system as 
\begin{align*}
 \theta_{\varphi}(P)=\nabla\varphi(P)
\end{align*}
Note that the matrix $\theta_\varphi(P)$ is not necessarily a positive definite matrix. 
Indeed, for the potential $\varphi(P)=-\log\det P$, we have $\theta_{\varphi}(P)=-P^{-1}$
which is a negative definite matrix. The function $\theta_\varphi$ is, however, one-to-one
mapping. Hence $\theta_\varphi(P)$ works as the coordinate system on $\PD(n)$. 
The inverse function of $\nabla\varphi$ is expressed by the conjugate function of
$\varphi$. The convex function $\varphi$ has the dual representation called Fenchel
conjugate, which is defined as 
\begin{align}
 \varphi^*(P)=\sup_{Q\in\PD(n)}\big\{\<P,Q\>-\varphi(Q)\big\}. 
 \label{eqn:Fenchel-dual}
\end{align}
Then, we have 
\begin{align*}
 \nabla\varphi^*(P)=(\nabla\varphi)^{-1}(P)=(\theta_\varphi)^{-1}(P) 
\end{align*}
on the domain of $\varphi^*$ \cite[Theorem 26.5]{rockafellar70:_convex_analy}. 
For any potential $\varphi$, the $\eta$-coordinate system is common and only the
$\theta_{\varphi}$-coordinate system depends on the potential. 

For the potential $V$ of the $V$-Bregman divergence, the $\theta_\varphi$-coordinate
system is denoted as $\theta_V(P)$, which is given as 
\begin{align*}
\theta_V(P)=-\nu(P)P^{-1}. 
\end{align*}
Thus $\theta_V(P)$ is a negative definite matrix for $P\in\PD(n)$. 

Let us define the flatness of a submanifold in $\mathrm{PD}(n)$. See \cite{AmariNagaoka00}
for the formal definition of the flatness with terminologies of differential geometry. 
\begin{definition}
 [autoparallel submanifold]
 Let $\mathcal{M}$ be a subset of $\mathrm{PD}(n)$. If $\mathcal{M}$ is represented as an
 affine subspace in the $\eta$-coordinate, then $\mathcal{M}$ is called
 $\eta$-autoparallel 
 submanifold. If $\mathcal{M}$ is represented as an affine subspace in the
 $\theta_\varphi$-coordinate, 
 then $\mathcal{M}$ is called $\theta_\varphi$-autoparallel submanifold. 
 When an $\eta$-autoparallel submanifold $\mathcal{M}$ is also
 $\theta_\varphi$-autoparallel, 
 $\mathcal{M}$ is called doubly autoparallel submanifold. 
\end{definition}
For the potential $\varphi(P)=V(\det{P})$, 
the $\theta_\varphi$-coordinate and the $\theta_\varphi$-autoparallel is 
denoted as 
the $\theta_V$-coordinate and the $\theta_V$-autoparallel, respectively. 
Formally, the flatness is defined from the connection on the differentiable manifold 
\cite{AmariNagaoka00,kobayashi96:_found_of_differ_geomet}. Here, we adopt a simplified
definition. 
\begin{example}
 \label{example:doubly-autoparallel-secant-cond}
 Let $V(z)$ be the negative logarithmic function $V(z)=-\log(z)$, 
 then we have $\nu(z)=1$. 
 The $\eta$-coordinate system is defined as $\eta(P)=P$, and the $\theta_V$-coordinate
 system is given as $\theta_V(P)=-P^{-1}$. 
 For two vectors $s,\,y\in\Real^n$ 
 we define the submanifold $\mathcal{M}$ which represents the secant condition such that 
 \begin{align*}
  \mathcal{M}
  =\{B\in\mathrm{PD}(n)~|~Bs=y\}. 
 \end{align*}
 Suppose $\mathcal{M}\neq \emptyset$, then we see that $\mathcal{M}$ is doubly
 autoparallel, since
 \begin{align*}
  \mathcal{M}
  =\{B\in\mathrm{PD}(n)\,|\,\eta(B)s=y\}
  =\{B\in\mathrm{PD}(n)\,|\,\theta_V(B)y=-s\}
 \end{align*}
 holds. That is, $\mathcal{M}$ is represented as the affine subspace in both the
 $\eta$-coordinate system and the $\theta_V$-coordinate system. 
\end{example}

\subsection{Extended Pythagorean Theorem}
\label{subsec:Extended_Pythagorean_Theorem}
The projection of a matrix in $\PD(n)$ onto an autoparallel submanifold is defined below. 
Then, we introduce the extended Pythagorean theorem. 
\begin{definition}[projection]
 Let $\varphi$ be a potential, $Q$ be a positive definite matrix.  
 An $\eta$-autoparallel submanifold in $\mathrm{PD}(n)$ is denoted as $\mathcal{M}$. 
 The matrix $P^*\in\mathcal{M}$ is called $\theta_\varphi$-projection of $Q$ onto $\mathcal{M}$, 
 when the equality 
 \begin{align*}
  \<\theta_\varphi(Q)-\theta_\varphi(P^*),\,\eta(P)-\eta(P^*)\>=0,\quad ^\forall P\in\mathcal{M}
 \end{align*}
 holds. 
 Let $\mathcal{N}$ be a $\theta_\varphi$-autoparallel submanifold in $\mathrm{PD}(n)$. 
 The matrix $P^*\in\mathcal{N}$ is called $\eta$-projection of $Q$ onto $\mathcal{N}$ when 
 the equality 
 \begin{align*}
  \<\eta(Q)-\eta(P^*),\,\theta_{\varphi}(P)-\theta_{\varphi}(P^*)\>=0,\quad ^\forall P\in\mathcal{N}
 \end{align*}
 holds. 
\end{definition}
Let $\mathcal{L}$ be a one-dimensional $\theta_\varphi$-autoparallel submanifold defined as
\begin{align*}
 \mathcal{L}=
 \big\{
 P\in\PD(n)\,|\,^\exists t\in\Real,\ 
 \theta_\varphi(P)=(1-t)\theta_{\varphi}(Q)+t\theta_{\varphi}(P^*)
 \big\}. 
\end{align*}
When $P^*$ is the $\theta_\varphi$-projection of $Q$ onto $\mathcal{M}$, 
the $\eta$-autoparallel submanifold $\mathcal{M}$ is orthogonal to 
$\mathcal{L}$ at $P^*$ with respect to the inner product $\<\cdot\,,\cdot\>$. 
In the $\eta$-projection, also the same picture holds by replacing $\eta$ and
$\theta_\varphi$. 

\begin{theorem}[Extended Pythagorean Theorem
 \cite{AmariNagaoka00,murata04:_infor_geomet_u_boost_bregm_diver}]
 \label{theorem:Extended_Pythagorean_Theorem}
 Let $\varphi$ be a potential function, $\mathcal{M}$ be an $\eta$-autoparallel submanifold in
 $\mathrm{PD}(n)$, and $Q$ be a positive definite matrix. 
 Then, the following three statements are equivalent.
\begin{description}
 \item[(a)] $P^*$ is a $\theta_\varphi$-projection of $Q$ onto $\mathcal{M}$. 
 \item[(b)] $P^*\in\mathcal{M}$ satisfies the equality
	    \begin{align}
	     \label{eqn:PythagoreanTheorem}
	     D_{\varphi}(P,Q) =  D_{\varphi}(P,P^*)+D_{\varphi}(P^*,Q)
	    \end{align}
	    for any $P\in\mathcal{M}$. 
 \item[(c)] $P^*$ is the unique optimal solution of the problem
	    \begin{align}
	     \label{eqn:opt-prob_V-projection_ontoM}
	     &\min_{P\in\mathrm{PD}(n)}\,D_{\varphi}(P,Q)\quad \subto P\in\mathcal{M}. 
	    \end{align}
\end{description}
\end{theorem}
\begin{proof}
For any $P,P^*,Q\in\mathrm{PD}(n)$ the equality 
\begin{align}
 \label{eqn:3dvi-relation}
 D_{\varphi}(P,Q)-D_{\varphi}(P,P^*)-D_{\varphi}(P^*,Q)
 =\<\theta_\varphi(Q)-\theta_\varphi(P^*),\,\eta(P^*)-\eta(P)\> 
\end{align}
holds. The equivalence between (a) and (b) follows the above equality. 
If (b) holds, then the non-negativity of the divergence assures that $P^*$ is an optimal
solution of \eqref{eqn:opt-prob_V-projection_ontoM}. The uniqueness follows the strict
convexity of the divergence $D_{\varphi}(P,Q)$ in $P$. Hence (c) holds. 
Finally, we show that (a) follows (c). 
Let $P^*$ be an optimal solution of \eqref{eqn:opt-prob_V-projection_ontoM}.  
The $\eta$-autoparallel submanifold $\mathcal{M}$ is represented by 
\begin{align*}
 \mathcal{M}
 &=\{P\in\mathrm{PD}(n)~|~\<\eta(P),A_i\>=b_i,\,i=1,\ldots,k\}
\end{align*} 
in which $A_i$ is an $n$ by $n$ real matrix and $b_i\in\Real$ for $i=1,\ldots,k$. 
The optimality condition of \eqref{eqn:opt-prob_V-projection_ontoM} yields that 
\begin{align*}
 -\theta_\varphi(P^*)+\theta_\varphi(Q)=\sum_{i=1}^{k}\lambda_i A_i,\quad \lambda_i\in\Real 
\end{align*}
with some $\lambda_1,\ldots,\lambda_k$. 
In addition, the fact that both $P$ and $P^*$ are included in $\mathcal{M}$ leads to the 
equalities 
\begin{align*}
 \<\eta(P^*)-\eta(P),A_i\>=0,\quad i=1,\ldots,k.
\end{align*}
Therefore, we obtain 
 \begin{align*}
  \<\theta_\varphi(Q)-\theta_\varphi(P^*),\eta(P^*)-\eta(P)\>=0
\end{align*}
 for any $P\in\mathcal{M}$. This implies that $P^*$ is a $\theta_{\varphi}$-projection of
 $Q$ onto $\mathcal{M}$. 
\end{proof}
The uniqueness of the $\theta_\varphi$-projection onto the $\eta$-autoparallel submanifold is 
shown through the equivalence between (a) and (b) in Theorem 
\ref{theorem:Extended_Pythagorean_Theorem}. 
The similar argument is valid for $\eta$-projection onto $\theta_\varphi$-autoparallel
 submanifold. We show the result without proof.  
\begin{theorem}
 \label{theorem:m-projection-Extended_Pythagorean_Theorem}
 Let $\varphi$ be a potential function, $\mathcal{N}$ be a $\theta_\varphi$-autoparallel
 submanifold in $\mathrm{PD}(n)$,  
 and $Q$ be a positive definite matrix. 
 Then, the following conditions (a) and (b) are equivalent.
\begin{description}
 \item[(a)] $P^*$ is an $\eta$-projection of $Q$ onto $\mathcal{N}$. 
 \item[(b)] $P^*\in\mathcal{N}$ satisfies the equality
	    \begin{align}
	     \label{eqn:m-proj-PythagoreanTheorem}
	     D_{\varphi}(Q,P) =  D_{\varphi}(Q,P^*)+D_{\varphi}(P^*,P)
	    \end{align}
	    for any $P\in\mathcal{N}$. 
\end{description}
 When (a) or (b) holds, $P^*$ is the unique optimal solution of the problem 
\begin{align}
 \label{eqn:min_Vdiv_st_Vautoparallel}
 &\min_{P\in\mathrm{PD}(n)}\,D_{\varphi}(Q,P)\quad \subto P\in\mathcal{N}. 
\end{align}
\end{theorem}
The Bregman divergence $D_{\varphi}(Q,P)$ may not be convex in $P$, and hence the conditions (a) or
(b) in Theorem \eqref{theorem:m-projection-Extended_Pythagorean_Theorem} 
is not necessarily derived from the optimality condition of 
\eqref{eqn:min_Vdiv_st_Vautoparallel}. 

As shown in Section \ref{sec:Introduction},  the BFGS/DFP update formulae are
derived by minimizing the KL-divergence. 
Example \ref{example:doubly-autoparallel-secant-cond} shows that the submanifold
associated with the secant condition $\mathcal{M}=\{B\in\mathrm{PD}(n)~|~Bs_k=y_k\}$ is
doubly autoparallel with respect to the flatness defined from the potential 
$V(z)=-\log z$. Thus, we obtain the following geometrical interpretation,  
\begin{description}
 \item[BFGS update:] $\theta_V$-projection of $B_k$ onto the $\eta$-autoparallel
	    submanifold $\mathcal{M}$, 
 \item[DFP update:] $\eta$-projection of $B_k$ onto the $\theta_V$-autoparallel
	    submanifold $\mathcal{M}$. 
\end{description}
Figure \ref{fig:dfp-bfgs-figure} presents the geometrical view of the standard
quasi-Newton updates based on information geometry. 
\begin{figure}[t]
 \begin{center}
  \scalebox{0.5}{\includegraphics{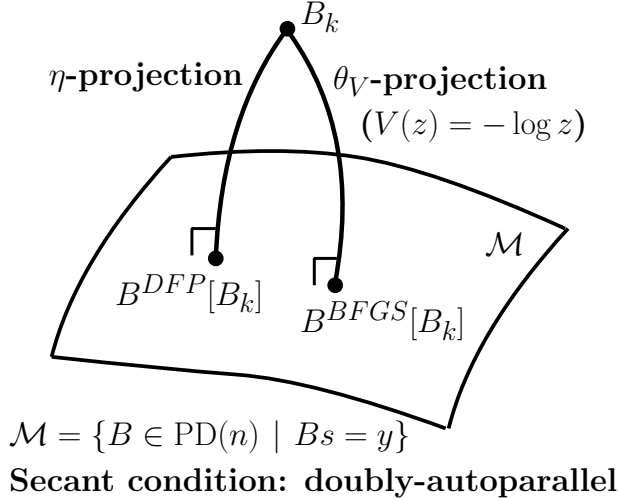}}
 \end{center}
 \vspace*{-5mm}
 \caption{
 Geometrical interpretation of quasi-Newton updates. 
 For the potential $V(z)=-\log z$, the submanifold $\mathcal{M}$ defined by the secant
 condition is doubly autoparallel with respect to $\eta$- and $\theta_V$-coordinate
 systems. 
 The BFGS formula $B^{\mathrm{BFGS}}[B_{k}]$ is given as the $\theta_V$-projection of $B_k$ onto the
 $\eta$-autoparallel submanifold $\mathcal{M}$, and the DFP update
 $B^{\mathrm{DFP}}[B_{k}]$ is given as the $\eta$-projection of $B_k$ onto the
 $\theta_V$-autoparallel submanifold $\mathcal{M}$. } 
 \label{fig:dfp-bfgs-figure}
\end{figure}

\section{quasi-Newton Methods based on Bregman Divergences}
\label{sec:potential-function_quasi-Newton}
We consider quasi-Newton update formulae derived from variational problems 
with respect to Bregman divergences. 
As shown in Section \ref{sec:Introduction}, the standard quasi-Newton updates are derived
from the minimization problem of the KL-divergence. We show that Bregman divergences lead
extended update formulae. In addition, an explicit expression of the extended Hessian
update formula is presented. 

We consider the minimization problem of the Bregman divergence 
instead of the KL-divergence. 
The extended BFGS update formula is given as the optimal solution of 
\begin{align}
 \label{eqn:Bregman-BFGS-hessian-update-prob}
 \min_{B\in\mathrm{PD}(n)}\ D_{\varphi}(B,B_k),\quad \text{subject to}\ \ Bs_k=y_k. 
\end{align}
Suppose that the optimal solution $B_{k+1}$ exists. 
Then $B_{k+1}$ is the unique $\theta_\varphi$-projection of $B_k$ onto the submanifold
defined from the secant condition. On the other hand, as the extension of the DFP update,
we consider the problem, 
\begin{align}
 \label{eqn:Bregman-DFP-hessian-update-prob}
 \min_{B\in\mathrm{PD}(n)}\ D_{\varphi}(B^{-1},B_k^{-1}),\quad \text{subject to}\ \ Bs_k=y_k. 
\end{align}
instead of the minimization of $\KL(B_k,B)=\KL(B^{-1},B_k^{-1})$. 
In the similar way, we can derive the quasi-Newton methods for the approximate inverse 
Hessian matrix $H_k=B_k^{-1}$. 

In the following we focus on the extension of the BFGS method
\eqref{eqn:Bregman-BFGS-hessian-update-prob}, since the same argument is valid for the 
extension of DFP method. 
A formal expression of the optimal solution is presented in the theorem below.  
\begin{theorem}
 Suppose that there exists an optimal solution
 \eqref{eqn:Bregman-BFGS-hessian-update-prob}. 
 Then the optimal solution $B_{k+1}$ is unique and satisfies 
 \begin{align*}
  B_{k+1}=
  \nabla\varphi^*\big(\nabla\varphi(B_k)+ s_k\lambda^\top+\lambda s_k^\top\big),\qquad
  B_{k+1}s_k=y_k, 
 \end{align*}
 where $\lambda\in\Real^n$ is a column vector and $\varphi^*$ is the Fenchel conjugate
 function of $\varphi$. 
\end{theorem}
\begin{proof}
 Since \eqref{eqn:Bregman-BFGS-hessian-update-prob} is a convex problem and 
 the objective function $D_{\varphi}(B,B_k)$ is strictly convex in $B$, we see that the optimal
 solution is unique if it exists. 
 Suppose that $B_{k+1}$ is the optimal solution of
 \eqref{eqn:Bregman-BFGS-hessian-update-prob}, then $B_{k+1}$ satisfies the optimality
 condition. According to G\"{u}ler, et al. \cite{guler09:_dualit_in_quasi_newton_method}, 
 the normal vector of the affine subspace 
 $\mathcal{M}=\{B\in\mathrm{PD}(n)~|~Bs_k=y_k\}$ 
 is characterized by the form of 
  \begin{align*}
  s_k\lambda^\top+\lambda s_k^\top\in\mathrm{Sym}(n),\qquad \lambda\in\Real^n. 
  \end{align*}
In fact for $B_1, B_2\in\mathcal{M}$ we have
\begin{align*}
  \<s_k\lambda^\top+\lambda s_k^\top,\,B_1-B_2\>
  &=
  \lambda^\top B_1 s_k+ s_k^\top B_1 \lambda
  -\lambda^\top B_2 s_k- s_k^\top B_2 \lambda\\
  &=
  \lambda^\top y_k+ y_k^\top \lambda
  -\lambda^\top y_k-y_k^\top \lambda\\
  &=0, 
\end{align*}
and thus $s_k\lambda^\top+\lambda s_k^\top$ is a normal vector of $\mathcal{M}$. 
G\"{u}ler, et al.~\cite{guler09:_dualit_in_quasi_newton_method} have shown that 
the normal vector is restricted to the expression above. 
Hence, for the optimal solution $B_{k+1}$ there exists $\lambda\in\Real^n$ such that 
$\nabla D_{\varphi}(B,B_k)\big|_{B=B_{k+1}}=s_k\lambda^\top+\lambda s_k^\top$
and $B_ks_k=y_k$ hold. The first equality is represented as 
$\nabla\varphi(B_{k+1})-\nabla\varphi(B_{k})=s_k\lambda^\top+\lambda s_k^\top$. 
The existence of $B_{k+1}$ assures that 
$B_{k+1}=\nabla\varphi^*\big(\nabla\varphi(B_{k})+s_k\lambda^\top+\lambda s_k^\top\big)$, 
where $\varphi^*$ is the Fenchel conjugate of $\varphi$ defined in \eqref{eqn:Fenchel-dual}. 
\end{proof}

For general Bregman divergences, we do not have the explicit expression of the Hessian
update formula. As a special case, we consider the minimization problem of the $V$-Bregman
divergence, 
\begin{align}
 \label{eqn:V-BFGS-hessian-update-prob}
 \text{$V$-BFGS:}\qquad 
 \min_{B\in\mathrm{PD}(n)}\ D_V(B,B_k),\quad \text{subject to}\ \ Bs_k=y_k. 
\end{align}
The update formula obtained by the problem above is referred to as the $V$-BFGS update
formula. The theorem below shows an explicit expression of the $V$-BFGS update formula. 
\begin{theorem}[$V$-BFGS update formula]
 \label{theorem:V-BFGS-form}
 Suppose the function $V$ is a potential function defined in 
 Definition \ref{def:V-Bregman-div_potential}. 
 Let $B_k\in\mathrm{PD}(n)$, and suppose $s_k^\top y_k>0$. 
 Then the problem \eqref{eqn:V-BFGS-hessian-update-prob} has the unique 
 optimal solution $B_{k+1}\in\mathrm{PD}(n)$ satisfying 
\begin{align}
 B_{k+1} &=
 \frac{\nu(\det{B_{k+1}})}{\nu(\det{B_k})}B^{BFGS}[B_k;s_k,y_k]
 +\bigg(1-\frac{\nu(\det{B_{k+1}})}{\nu(\det{B_k})}\bigg) 
 \frac{y_ky_k^\top}{s_k^\top y_k}. 
 \label{eqn:update-formula-V-BFGS}
\end{align}
\end{theorem}
Though the theorem is proved in \cite{kanamori10:_bregm_exten_of_quasi_newton_updat_ii}, 
the proof is also found in Appendix \ref{appendix:VBFGS-formula} of the present paper 
as a supplementary. 
In the same way, we can obtain the explicit formula of the $V$-DFP update formula, which
is the minimizer of $D_V(B^{-1},B_k^{-1})$ subject to $Bs_k=y_k$. 
The update formula is equivalent to the self-scaling quasi-Newton update 
defined as 
\begin{align}
 \label{eqn:self-scaling}
 B_{k+1} &=\theta_k B^{BFGS}[B_k;s_k,y_k]
 +(1-\theta_k) \frac{y_ky_k^\top}{s_k^\top y_k}, 
\end{align}
where $\theta_k$ is a positive real number. 
Various choices for $\theta_k$ have been proposed, see
\cite{oren74:_self_scalin_variab_metric_ssvm,nocedal93:_analy_of_self_scalin_quasi_newton_method}. 
A popular choice is $\theta_k=s_k^\top y_k/s_k^\top B_ks_k$. 
In the $V$-BFGS update formula, the
coefficient $\theta_k$ is determined from the function $\nu$. 

We present a practical way of computing the Hessian approximation
\eqref{eqn:update-formula-V-BFGS}. Details are shown in the sequel 
\cite{kanamori10:_bregm_exten_of_quasi_newton_updat_ii}. 
In Eq~\eqref{eqn:update-formula-V-BFGS}, the optimal solution $B_{k+1}$ appears in both 
sides, that is, we have only the implicit expression of $B_{k+1}$. 
The numerical computation is, however, efficiently conducted as well as the standard BFGS 
update. To compute the matrix $B_{k+1}$, first we compute the determinant $\det B_{k+1}$.  
The determinant of both sides of \eqref{eqn:update-formula-V-BFGS} leads to 
\begin{align}
 \det B_{k+1}=\frac{\det(B^{BFGS}[B_k;s_k,y_k])}{\nu(\det B_{k})^{n-1}}\cdot\nu(\det
 B_{k+1})^{n-1}. 
 \label{eqn:determinant-V-BFGS}
\end{align}
Hence, by solving the nonlinear equation 
\begin{align*}
 z=\frac{\det(B^{BFGS}[B_k;s_k,y_k])}{\nu(\det B_{k})^{n-1}}\cdot\nu(z)^{n-1},\qquad z>0
\end{align*}
we can find $\det B_{k+1}$. 
As shown in the proof of Theorem \ref{theorem:V-BFGS-form}, the function $z/\nu(z)^{n-1}$
is monotone increasing. Hence the Newton method is available to find the root of the above
equation efficiently. 
Once we obtain the value of $\det B_{k+1}$, we can compute the Hessian approximation $B_{k+1}$ by
substituting $\det B_{k+1}$ into Eq~\eqref{eqn:update-formula-V-BFGS}. 
Figure \ref{fig:V-BFGS-update} shows the update algorithm of the $V$-BFGS formula which
exploits the Cholesky decomposition of the approximate Hessian matrix. 
By maintaining the Cholesky decomposition, we can easily compute the the determinant and
the search direction. 
The convergence property of the quasi-Newton method with the $V$-BFGS update formula is
considered in \cite{kanamori10:_bregm_exten_of_quasi_newton_updat_ii}. 

\begin{figure}[p]
 \centering 
 \fbox{
 \begin{minipage}{0.9\linewidth}
\begin{description}
 \item[$V$-BFGS update:] 
 \item[Initialization:] 
            The function $\nu(z)$ denotes $-V'(z)z$. 
            Let $B_0\in\mathrm{PD}(n)$ be a matrix which is an initial
            approximation of the Hessian matrix, and $L_0L_0^\top=B_0$ be 
            the Cholesky decomposition of $B_0$. Let $x_0\in\Real^n$ be an initial point,
            and set $k=0$. 
 \item[Repeat:] If stopping criterion is satisfied, go to Output. 
            \begin{enumerate}
             \item Let $x_{k+1}=x_k-\alpha_k B_k^{-1}\nabla f(x_k)$, 
		   where $\alpha_k\geq 0$ is a step length satisfying the Wolfe condition
		   \cite[Section 3.1]{nocedal99:_numer_optim}. 
		   The Cholesky decomposition $B_k=L_kL_k^\top$ is available to compute
                   $B_k^{-1}\nabla f(x_k)$. 
	     \item Set $s_k=x_{k+1}-x_k$ and $y_k=\nabla f(x_{k+1})-\nabla f(x_k)$. 
             \item Update $L_k$ to $\bar{L}$ which is the Cholesky decomposition 
		   of $B^{BFGS}[B_k;s_k,y_k]$, that is, 
                   \begin{align*}
                    \bar{L}\bar{L}^\top
		    =B^{BFGS}[B_k;s_k,y_k]=B^{BFGS}[L_kL_k^\top;s_k,y_k]. 
                   \end{align*}
                   The Cholesky decomposition with rank-one update is available. 
             \item Compute 
		   \begin{align*}
		    C=\frac{(\det{\bar{L}})^2}{\nu((\det{L_k})^2)^{n-1}}
		   \end{align*}
		   and find the root of the equation 
                   \begin{align*}
                    C\cdot\nu(z)^{n-1} = z,\qquad z>0. 
                   \end{align*}
                   Let the solution be $z^*$. 
             \item Compute the Cholesky decomposition $L_{k+1}$ such that 
                   \begin{align*}
		    L_{k+1}L_{k+1}^\top = 
		    \frac{\nu(z^{*})}{\nu((\det{L_k})^2)}
		    \bar{L}\bar{L}^\top 
                    +\bigg(1-\frac{\nu(z^{*})}{\nu((\det{L_k})^2)}\bigg) 
		    \frac{y_ky_k^\top}{s_k^\top y_k}. 
                   \end{align*}
             \item $k\leftarrow k+1$. 
            \end{enumerate}
 \item[Output:]  Local optimal solution $x_{k}$. 
\end{description}
 \end{minipage}}
 \caption{Pseudo code of $V$-BFGS method. 
 The Cholesky decomposition with rank-one update is useful in the algorithm. }
 \label{fig:V-BFGS-update}
\end{figure}

\begin{example}
 \label{example:VBFGS-power-div}
 We show the $V$-BFGS formula derived from the power potential. 
 Let $V(z)$ be the power potential $V(z)=(1-z^\gamma)/\gamma$ with $\gamma<1/n$. 
 As shown in Example \ref{example:power-div}, we have $\nu(z)=z^\gamma$. 
 Due to the equality 
 \begin{align*}
  \det(B^{BFGS}[B_k;s_k,y_k])=\det(B_k)\frac{s_k^\top y_k}{s_k^\top B_ks_k}
 \end{align*}
 and Eq.~\eqref{eqn:determinant-V-BFGS}, we have 
 \begin{align*}
  \frac{\nu(\det{B_{k+1}})}{\nu(\det{B_k})}=
  \left(\frac{s_k^\top y_k}{s_k^\top B_ks_k}\right)^\rho,\qquad
  \rho=\frac{\gamma}{1-(n-1)\gamma}. 
 \end{align*}
 Then the $V$-BFGS update formula is given as 
\begin{align*}
 B_{k+1} =
\left(\frac{s_k^\top y_k}{s_k^\top B_ks_k}\right)^\rho
 B^{BFGS}[B_k;s_k,y_k]
 +\bigg(1-\left(\frac{s_k^\top y_k}{s_k^\top B_ks_k}\right)^\rho
 \bigg)
 \frac{y_ky_k^\top}{s_k^\top y_k}. 
\end{align*}
 For $\gamma$ such that $\gamma<1/n$, we have $-1/(n-1)<\rho<1$. 
 In the standard self-scaling update formula \eqref{eqn:self-scaling}, 
 the above matrix $B_{k+1}$ with $\rho=1$ is used, 
 while it is not derived from the strictly convex potential function. 
\end{example}

\section{Invariance of Update Formulae under Group Action}
\label{sec:Invariance}
In this section we study the invariance of the $V$-BFGS update formula
\eqref{eqn:update-formula-V-BFGS} under the affine coordinate transformation of the
optimization variable. 
For the minimization problem of the function $f(x)$, let us consider the variable change
of $x$. For a non-degenerate matrix $T\in\mathrm{GL}(n)$, the variable change is defined
by 
\begin{align}
 x=T^{-1}\widetilde{x}, 
 \label{eqn:variable_change}
\end{align}
then the function $f(x)$ is transformed to $\widetilde{f}(\widetilde{x})$ defined as 
\begin{align*}
 \widetilde{f}(\widetilde{x})=f(T^{-1}\widetilde{x}). 
\end{align*}
Then we have 
\begin{align*}
 \nabla\widetilde{f}(\widetilde{x})=(T^\top)^{-1}\nabla f(T^{-1}\widetilde{x}),\qquad
 \nabla^2\widetilde{f}(\widetilde{x})=(T^\top)^{-1}(\nabla^2 f(T^{-1}\widetilde{x})) T^{-1}. 
\end{align*}
Our concern is how the point sequence $\{x_k\}_{k=1}^\infty$ generated by the $V$-BFGS
method is transformed by the variable change \eqref{eqn:variable_change}. 

We consider the Hessian approximation matrix under the variable change. Let
$B_k\in\mathrm{PD}(n)$ be the Hessian approximation computed at the $k$-th step of the
$V$-BFGS update for the minimization of $f(x)$. We now define  
\begin{align*} 
 \widetilde{x}_k = Tx_k,\qquad \widetilde{B}_k = (T^\top)^{-1}B_kT^{-1}. 
\end{align*}
Let $\widetilde{B}_{k+1}$ be the Hessian approximation matrix updated from 
$\widetilde{B}_k$ for the function $\widetilde{f}(\widetilde{x})$, where we suppose that
the $V$-BFGS method is used for the minimization of $\widetilde{f}(\widetilde{x})$. 
We consider the relation between $B_{k+1}$ and $\widetilde{B}_{k+1}$. 
The updated point $\widetilde{x}_{k+1}$ 
is determined by 
\begin{align*}
 \widetilde{x}_{k+1}= \widetilde{x}_{k}-
 \widetilde{\alpha}_k \widetilde{B}_k^{-1}\nabla \widetilde{f}(\widetilde{x}_k), 
\end{align*}
where $\widetilde{\alpha}_k$ is a non-negative real number determined 
by a line search. Then we have 
\begin{align}
 \widetilde{f}(\widetilde{x}_{k}-
 \widetilde{\alpha}_k \widetilde{B}_k^{-1}\nabla \widetilde{f}(\widetilde{x}_k))
 =
 \widetilde{f}(T(x_k-\widetilde{\alpha}_k B_k^{-1}\nabla f(T^{-1}\widetilde{x}_k)))
 =
 f(x_k-\widetilde{\alpha}_k B_k^{-1}\nabla f(x_k)). 
 \label{eqn:step-length-equality}
\end{align}
Let $\alpha_k$ be the step length for the function $f(x)$ at the $k$-th step of the
$V$-BFGS method. 
Due to the equality \eqref{eqn:step-length-equality}, we see that the step length
$\widetilde{\alpha}_k$ is identical to $\alpha_k$, if the line search with the same
stopping rule is applied for both $f(x)$ and $\widetilde{f}(\widetilde{x})$. 
As the result, the equality $\widetilde{x}_{k+1}=Tx_{k+1}$ holds under the condition
$\alpha_k=\widetilde{\alpha}_k$. Let $\widetilde{s}_k$ and $\widetilde{y}_k$ be 
\begin{align*}
 \widetilde{s}_k=\widetilde{x}_{k+1}-\widetilde{x}_{k},\qquad
 \widetilde{y}_k=\nabla\widetilde{f}(\widetilde{x}_{k+1})-\nabla\widetilde{f}(\widetilde{x}_{k})
\end{align*}
then we obtain the equalities, 
\begin{align*}
 \widetilde{s}_k=Ts_k,\qquad  \widetilde{y}_k=(T^\top)^{-1}y_k. 
\end{align*}

\begin{figure}[t]
 \begin{center}
  \scalebox{0.6}{\includegraphics{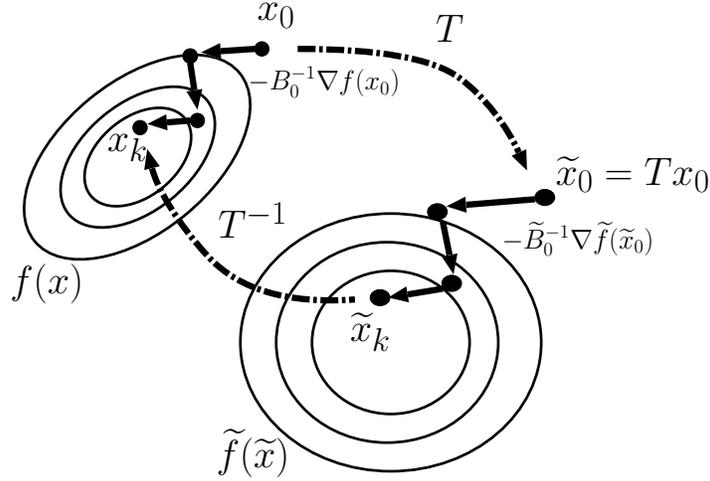}}
 \end{center}
 \vspace*{-5mm}
 \caption{The coordinate transformation between $x$ of the function $f$ and
 $\widetilde{x}$ of the function $\widetilde{f}$ is depicted. The initial point $x_0$ is
 transformed to 
 $\widetilde{x}_0=Tx_0$ and the search direction at $x_0$ is also transformed to
 $-\widetilde{B}_0^{-1}\nabla\widetilde{f}(\widetilde{x}_0)$. The quasi-Newton method is
 applied to both $f(x)$ and $\widetilde{f}(\widetilde{x})$, and then 
 the points $x_k$ and $\widetilde{x}_k$ are obtained in each coordinate system. 
 If the equality $T^{-1}\widetilde{x}_k=x_k$ holds, 
 the optimization algorithm is invariant under the transformation with $T$. } 
 \label{fig:invariance}
\end{figure}

We consider the condition of $T$ such that the equality 
\[
T^\top\widetilde{B}_{k+1}T=B_{k+1},
\]
holds, when $\widetilde{x}_k=Tx_k$ and $\widetilde{B}_k=(T^\top)^{-1}B_k T^{-1}$ are
satisfied. 
For such $T$, the equality $\widetilde{x}_{k+1}=T x_{k+1}$ recursively holds. 
This implies that the point sequence obtained by the $V$-BFGS method is invariant under
the affine transformation \eqref{eqn:variable_change}. In the optimization of
$\widetilde{f}(\widetilde{x})$ by the $V$-BFGS method, the matrix $\widetilde{B}_k$ is
updated to $\widetilde{B}_{k+1}$ such that 
\begin{align*}
 \widetilde{B}_{k+1}
 &=
 \frac{\nu(\det{\widetilde{B}_{k+1}})}{\nu(\det{\widetilde{B}_k})}
 B^{BFGS}[\widetilde{B}_k;\widetilde{s}_k,\widetilde{y}_k]
 +
 \bigg(1-\frac{\nu(\det{\widetilde{B}_{k+1}})}{\nu(\det{\widetilde{B}_k})}\bigg)
 \frac{\widetilde{y}_k\widetilde{y}_k^\top}{\widetilde{s}_k^\top \widetilde{y}_k}. 
\end{align*}
Some calculation yields that 
\begin{align}
 T^\top \widetilde{B}_{k+1}T
 &=
 \frac{\nu(\det{\widetilde{B}_{k+1}})}{\nu(\det{\widetilde{B}_k})}
 B^{BFGS}[B_k;s_k,y_k]
 +
 \bigg(1-\frac{\nu(\det{\widetilde{B}_{k+1}})}{\nu(\det{\widetilde{B}_k})}\bigg)
 \frac{y_ky_k^\top}{s_k^\top y_k}. 
\label{eqn:T-group-action-B-}
\end{align}
The following theorem provides a sufficient condition on $T$ such that 
$T^\top\widetilde{B}_{k+1}T=B_{k+1}$ holds. 
\begin{theorem}
 Suppose that $T\in\mathrm{SL}(n)$, that is, $\det(T)=1$. 
 Then the equality $T^\top\widetilde{B}_{k+1}T=B_{k+1}$ holds for any $V$-BFGS update
 formula. 
\end{theorem}
\begin{proof}
 Due to the assumption $\det(T)=1$, we have $\det(B_k)=\det(\widetilde{B}_k)$. Then
 Eq.\eqref{eqn:T-group-action-B-} is equivalent with 
 \begin{align*}
 T^\top \widetilde{B}_{k+1}T
 &=
 \frac{\nu(\det{\widetilde{B}_{k+1}})}{\nu(\det{B_k})}B^{BFGS}[B_k;s_k,y_k]
 +\bigg(1-\frac{\nu(\det{\widetilde{B}_{k+1}})}{\nu(\det{B_k})}\bigg)
 \frac{y_ky_k^\top}{s_k^\top y_k}. 
 \end{align*}
 Hence, the determinant of $T^\top \widetilde{B}_{k+1}T$ yields the equality 
\begin{align*}
 \frac{\det(\widetilde{B}_{k+1})}{\nu(\det{\widetilde{B}_{k+1}})^{n-1}}
 &=
 \frac{\det\big(B^{BFGS}[B_k;s_k,y_k]\big)}{\nu(\det{B_k})^{n-1}}, 
\end{align*}
 where $\det(T^\top\widetilde{B}_{k+1}T)=\det{\widetilde{B}_{k+1}}$ is used. 
On the other hand, the matrix $B_{k+1}$ defined by the $V$-BFGS update formula
 \eqref{eqn:update-formula-V-BFGS} also satisfies, 
\begin{align*}
 \frac{\det(B_{k+1})}{\nu(\det{B_{k+1}})^{n-1}}
 &=
 \frac{\det\big(B^{BFGS}[B_k;s_k,y_k]\big)}{\nu(\det{B_k})^{n-1}}, 
\end{align*}
As shown in the proof of Theorem \ref{theorem:V-BFGS-form}, the function $z/\nu(z)^{n-1}$
is one to one mapping, and thus we have $\det{\widetilde{B}_{k+1}}=\det{B_{k+1}}$. 
Therefore, the equality $T^\top \widetilde{B}_{k+1}T=B_{k+1}$ holds. 
\end{proof}

Next, we study the variable change with $T\in\mathrm{GL}(n)$. Below we assume $\nu(1)=1$
without loss of generality. Let us define 
\begin{align*}
 b_k=\det{B_k},\quad b_{k+1}=\det{B_{k+1}},\quad
 \widetilde{b}_{k+1}=\det{\widetilde{B}_{k+1}},\quad  
 t=\det{T}
\end{align*}
and 
\begin{align*}
 a=\frac{\det{B^{BFGS}[B_k;s_k,y_k]}}{\nu(\det{B_k})^{n-1}}. 
\end{align*}
In the $V$-BFGS update formula, the determinant of $B_{k+1}$  leads the equality 
\begin{align}
 b_{k+1}=a\cdot\nu(b_{k+1})^{n-1}. 
 \label{eqn:determinant-B}
\end{align}
 The matrix $\widetilde{B}_{k+1}$ satisfies the update formula 
 \eqref{eqn:T-group-action-B-}, thus the determinant of both sides yields the equality 
\begin{align}
 \label{eqn:determinant-Btilde}
 \widetilde{b}_{k+1}\,t^2
 =a\cdot\left(\frac{\nu(\widetilde{b}_{k+1})\nu(b_k)}{\nu(b_kt^{-2})}\right)^{n-1}. 
\end{align}
When $T^\top \widetilde{B}_{k+1}T=B_{k+1}$ holds, Eq.\eqref{eqn:determinant-Btilde} 
is represented as 
\begin{align}
 \label{eqn:determinant-Btilde-invariant}
 b_{k+1}=a\cdot\left(\frac{\nu(b_{k+1}t^{-2})\nu(b_k)}{\nu(b_kt^{-2})}\right)^{n-1}. 
\end{align}

We consider the function $\nu$ which satisfies 
\eqref{eqn:determinant-B} and \eqref{eqn:determinant-Btilde-invariant} simultaneously. 
For a positive number $a>0$, let $b_a$ be the unique solution of the equation of $b$, 
\begin{align*}
 b=a\cdot\nu(b)^{n-1},\qquad b>0, 
\end{align*}
and $E_\nu=\{b_a\in\Real~|~a>0\}$ be the set of all possible solutions of the above
equation. 
Note that $1\in E_\nu$ holds for any $\nu$ since $1=1\cdot\nu(1)^{n-1}$ holds. 
\begin{theorem}
 \label{theorem:invariance-determinant}
 Let $\nu(z)>0$ be a differentiable function on $\Real_+$. 
 Suppose that there exists an open subset
 $E\subset\Real$ satisfying $1\in E\subset E_\nu$. 
 For the Hessian approximation by the $V$-BFGS method, 
 suppose that the equality 
 $\widetilde{B}_{k+1}=(T^\top)^{-1}B_{k+1}T^{-1}$ holds for all 
 $T\in\mathrm{GL}(n)$, all $B_k\in\PD(n)$ and all $s_k, y_k\in\Real^n$ satisfying
 $s_k^\top y_k>0$. 
 Then the function $\nu$ is equal to $\nu(z)=z^\gamma$ with some $\gamma\in\Real$. 
\end{theorem}
Note that $E_\nu=\Real_+$ holds for $\nu(z)=z^\gamma$ unless $\gamma=1/(n-1)$. 
\begin{proof}
 Under the assumption, 
 the equations \eqref{eqn:determinant-B} and
 \eqref{eqn:determinant-Btilde-invariant} share the same solution $b_{k+1}$ for any
 $a>0,\, b_k>0$ and $t\neq 0$. 
 Let $b_k=1,\,x=t^{-2}>0$. For any positive $a$ and $x$, equations 
 \eqref{eqn:determinant-B} and \eqref{eqn:determinant-Btilde-invariant} lead to 
\begin{align*}
 b_a=a\cdot\nu(b_a)^{n-1}\ \ \text{and}\ \  
 b_a
 =a\cdot\nu(b_a x)^{n-1}\left(\frac{\nu(1)}{\nu(x)}\right)^{n-1}
 =a\cdot\frac{\nu(b_a x)^{n-1}}{\nu(x)^{n-1}} 
\end{align*} 
 for $b_a\in E_\nu$. Hence we obtain 
\begin{align}
 \nu(b_a x)=\nu(b_a)\nu(x),\ \ a>0,\ x>0
 \ \Longleftrightarrow\ 
 \nu(b x)=\nu(b)\nu(x),\ \ b\in E_\nu,\ x>0. 
 \label{eqn:nu-isomorphism}
\end{align}
 The assumption on $E_\nu$ guarantees that 
 $1+\varepsilon\in E_\nu$ holds for any infinitesimal $\varepsilon$. 
 Thus Eq.\eqref{eqn:nu-isomorphism} leads the following expression, 
\begin{align*}
 \frac{\nu(x (1+\varepsilon))-\nu(x)}{x\varepsilon}
 =\frac{\nu(x)}{x}\cdot\frac{\nu(1+\varepsilon)-\nu(1)}{\varepsilon}. 
\end{align*}
Taking the limit $\varepsilon\rightarrow0$, we obtain the differential equation, 
\begin{align*}
\nu'(x)=\nu'(1)\frac{\nu(x)}{x}, \qquad \nu(1)=1, 
\end{align*}
and the solution is given as $\nu(x)=x^{\nu'(1)}$. 
\end{proof}
As shown in Example \ref{example:power-div}, the function $\nu(z)=z^\gamma$ is derived 
from the power potential $V(z)=(1-z^\gamma)/\gamma$. 
In robust statistics, the power potential has been applied in wide-rage of data analysis
\cite{basu98:_robus_and_effic_estim_by,minami02:_robus_blind_sourc_separ_by_beta_diver}. 

\begin{remark}
 Ohara and Eguchi \cite{ohara05:_geomet_posit_defin_matric_and} have studied 
 the differential geometrical structure over $\mathrm{PD}(n)$ induced by the $V$-Bregman
 divergence. 
 They pointed out that the geometrical structure is invariant under $\mathrm{SL}(n)$ group
 action. Furthermore, they have showed that 
 for the power potential $V(z)=(1-z^\gamma)/\gamma$, 
 the $\theta_V$- ($\eta$-) projection onto $\eta$- ($\theta_V$-) autoparallel submanifold 
 is invariant under $\mathrm{GL}(n)$ group action. 
 It turns out that only the orthogonality is kept unchanged under the group action. 
 The other geometrical features such as angle between two tangent vectors are not
 preserved in general. 
 Theorem \ref{theorem:invariance-determinant} indicates that
 the invariance of the geometrical structure on $\mathrm{PD}(n)$ is inherited 
 to the invariance of point sequences of quasi-Newton methods under the affine
 transformation. 
\end{remark}

In summary, we obtain the following results. 
Suppose that $\widetilde{x}_0=Tx_0,\, \widetilde{B}_0=(T^\top)^{-1}B_0T^{-1}$ holds. 
Let $\{x_k\}$ and $\{\widetilde{x}_k\}$ be point sequences generated by the $V$-BFGS
method for the functions $f(x)$ and $\widetilde{f}(\widetilde{x})$, respectively. 
Suppose that the line search with the same stopping rule is used for the step length. 
Then, for any $T\in\mathrm{SL}(n)$ the equality $\widetilde{x}_k=Tx_k$ holds for all
$k\geq 1$. 
Moreover the equality $\widetilde{x}_k=Tx_k,\,k\geq 1$ holds for any $T\in\mathrm{GL}(n)$
if and only if the function $V(z)$ is the power potential. 


\section{Geometry of Sparse quasi-Newton updates}
\label{sec:Sparse-V-quasi-Newton}
Sparse quasi-Newton method exploits the sparsity of Hessian matrix 
in order to reduce the computation cost
\cite{yamashita08:_spars_quasi_newton_updat_with}. 
The sparsity pattern of the Hessian matrix at a point $x\in\Real^n$ 
is represented by an index set $F$ satisfying 
\begin{align*}
 \{(i,j)~|~(\nabla^2 f(x))_{ij}\neq 0\}\subset F. 
\end{align*}
When the number of entries in $F$ is small, the matrix $\nabla^2 f(x)$ is referred to as 
sparse matrix. We assume that $(j,i)\in F$ holds for $(i,j)\in F$ and that $(i,i)\in F$
for all $i=1,\ldots,n$. 
Given a sparsity pattern $F$, the set of sparse matrix is defined by
\begin{align*}
 \mathcal{S}=\{P\in\mathrm{PD}(n)~|~P_{ij}=0\  \text{for}\  (i,j)\not \in F\}. 
\end{align*}
Clearly the submanifold $\mathcal{S}$ is $\eta$-autoparallel in $\mathrm{PD}(n)$. 

Yamashita \cite{yamashita08:_spars_quasi_newton_updat_with} has proposed a sparse
quasi-Newton method. In this section we show an extension of sparse quasi-Newton method 
and illustrate a geometrical structure of the update formula. 
First, we briefly introduce the sparse quasi-Newton method proposed by Yamashita 
\cite{yamashita08:_spars_quasi_newton_updat_with}. 
Suppose $H_k$ be an approximate inverse Hessian matrix at the $k$-th step of the sparse
quasi-Newton method. 
Let $H_k^{\mathrm{QN}}$ be the updated matrix of $H_k$ by the existing quasi-Newton
methods such as the BFGS or the DFP method for the approximate inverse Hessian matrix. 
In the computation of $H_k^{\mathrm{QN}}$, we need only the elements
$(H_k^{\mathrm{QN}})_{ij}$ for $(i,j)\in F$, and thus efficient computation will be
possible even if the size of the matrix is large. 
Then, compute the sparse matrix $H_{k+1}\in\mathcal{S}$ satisfying 
the constraint $(H_{k+1})_{ij}=(H_k^{\mathrm{QN}})_{ij}$ for all $(i,j)\in{F}$. 
The calculation of $H_{k+1}$ from $H_k^{\mathrm{QN}}$ is regarded as the
$\theta_V$-projection with respect to the KL-divergence. 
The sparse clique-factorization technique 
\cite{fukuda00:_exploit_spars_in_semid_progr,grone84:_posit_defin_compl_of_partial_hermit_matric}
is available for the practical computation of the projection. 
See \cite{yamashita08:_spars_quasi_newton_updat_with} for details.

For the computation of both $H^{\mathrm{QN}}_{k+1}$ and $H_{k+1}$ in the sparse
quasi-Newton method, we can use Bregman divergence instead of the KL-divergence. 
Figure \ref{fig:Extendedsparse_quasi-Newton} shows an extended sparse quasi-Newton method
for the approximate Hessian matrix $B_k$. 
Figure \ref{fig:geometry_Extendedsparse_quasi-Newton} illustrates the geometrical
interpretation of the extended sparse quasi-Newton updates. 

We have some choices in the algorithm of Figure \ref{fig:Extendedsparse_quasi-Newton}: 
(i) the Bregman divergence in Step 2, (ii) projection in Step 3, and (iii) the number of $T$. 
In the sparse quasi-Newton updates presented by Yamashita
\cite{yamashita08:_spars_quasi_newton_updat_with} ,
the number of iteration is set to $T=1$; 
in Step 2, the standard BFGS/DFP method for the approximate inverse Hessian is used; in
Step 3 the $\theta_V$-projection defined from the KL-divergence is computed. 
Moreover, the superlinear convergence has been proved, see
\cite{yamashita08:_spars_quasi_newton_updat_with} for details. 
In the following, we present the geometrical interpretation of the sparse quasi-Newton
method. 
Then we show a computation algorithm for the update formula derived from the $V$-Bregman
divergence. 

\begin{figure}[t]
 \begin{center}
 \fbox{
  \begin{minipage}{0.9\linewidth}
\begin{description}
 \item[Extended sparse quasi-Newton update algorithm:]
	    the Hessian approximation $B_k$ at the $k$-th step of quasi-Newton method is
	    updated to a sparse matrix $B_{k+1}$. 
	    Suppose that the Bregman divergence is defined from the potential function
	    $\varphi$, and let $\mathcal{S}$ be the set of sparse matrix defined by 
	    a fixed index set $F$. 
 \item[Initialization:] 
	    Let $T$ be an positive integer, and $B^{(0)}:=B_k$. 
 \item[Repeat:] $t=1,2,\ldots,T$. 
	    \begin{enumerate}
	     \item Compute the partial matrix $\bar{B}^{(t-1)}_{ij}$ for $(i,j)\in F$ from
		   $B^{(t-1)}$ by using the extended quasi-Newton method such as 
		   \eqref{eqn:Bregman-BFGS-hessian-update-prob} 
		   or \eqref{eqn:Bregman-DFP-hessian-update-prob}. 
	     \item Compute the sparse matrix $B^{(t)}\in\mathcal{S}$ which is the
		   $\theta_\varphi$-projection of $\bar{B}^{(t-1)}$ onto $\mathcal{S}$. 
	    \end{enumerate}
 \item[Output:] The updated approximate Hessian matrix $B_{k+1}$ is given as
	    $B^{(T)}\in\mathcal{S}$. 
\end{description}
\end{minipage}}
 \caption{An extension of sparse quasi-Newton method is presented. 
  The approximate Hessian $B_k$ is updated to $B_{k+1}$ by exploiting the update formula
  with Bregman divergences.}
\label{fig:Extendedsparse_quasi-Newton}
\end{center}
\vspace*{3mm}
\end{figure}

\begin{figure}[t]
\begin{center}
 \includegraphics[scale=0.5]{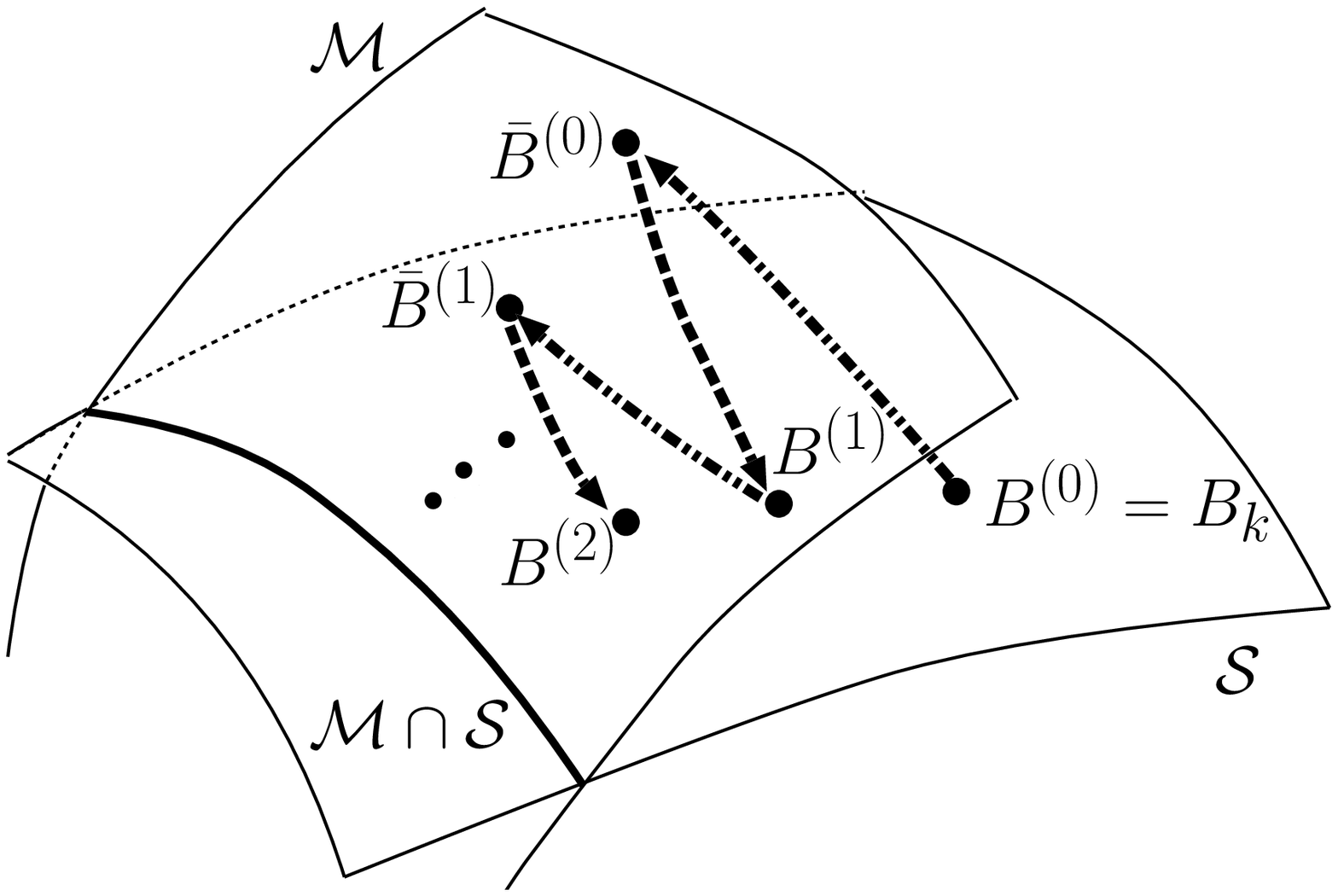}
 \caption{Geometrical illustration of the extended sparse quasi-Newton update algorithm.}
 \label{fig:geometry_Extendedsparse_quasi-Newton} 
 \end{center}
\end{figure}

\subsection{Geometry of Sparse quasi-Newton update}
\label{subsec:Geometry_Sparse-quasi-Newton}
We consider the sparse quasi-Newton update formula from the geometrical viewpoint. 
Remember that $\mathcal{M}$ is the set of matrices satisfying the secant condition 
\begin{align*}
\mathcal{M}=\{B\in\mathrm{PD}(n)~|~Bs_k=y_k\}. 
\end{align*}
Below we consider two kinds of update formulae:
\begin{description}
 \item[Algorithm 1: ] 
	    In the algorithm in Figure \ref{fig:Extendedsparse_quasi-Newton}, 
	    the matrix $\bar{B}^{(t)}$ is defined as the $\eta$-projection of $B^{(t)}$ onto 
	    $\mathcal{M}$, that is, $\bar{B}^{(t)}$ is equal to $B^{DFP}[B^{(t)};s_k,t_k]$. 
	    Then $B^{(t+1)}$ is defined as the $\theta_\varphi$-projection of
	    $\bar{B}^{(t)}$ onto $\mathcal{S}$. 
 \item[Algorithm 2: ] 
	    In the algorithm in Figure \ref{fig:Extendedsparse_quasi-Newton}, 
	    the matrix $\bar{B}^{(t)}$ is the $\theta_\varphi$-projection of $B^{(t)}$
	    onto $\mathcal{M}$, that is, $\bar{B}^{(t)}$ is given as the optimal solution
	    of \eqref{eqn:Bregman-BFGS-hessian-update-prob}. 
	    Then $B^{(t+1)}$ is defined as the $\theta_\varphi$-projection of
	    $\bar{B}^{(t)}$ onto $\mathcal{S}$. 
\end{description}
The difference between Algorithm 1 and Algorithm 2 is the projection onto $\mathcal{M}$ to
obtain $\bar{B}^{(t)}$. Below we show the theoretical properties for each algorithm. 

In Algorithm 1, we consider how the Bregman divergence $D_\varphi(B^{(t)},\bar{B}^{(t)})$
is updated. Let $B^{(0)}=B_k\in\mathcal{S}$ and suppose that the
$\theta_\varphi$-projection onto $\mathcal{S}$ exists. Then, the extended Pythagorean
theorem in Section \ref{subsec:Extended_Pythagorean_Theorem} leads that 
\begin{align*}
 D_\varphi(B^{(t)},\bar{B}^{(t)}) 
 &=
 D_\varphi(B^{(t)},B^{(t+1)})+D_\varphi(B^{(t+1)},\bar{B}^{(t)})\\
 &=
 D_\varphi(B^{(t)},B^{(t+1)})
 +D_\varphi(B^{(t+1)},\bar{B}^{(t+1)})
 +D_\varphi(\bar{B}^{(t+1)},\bar{B}^{(t)})\\
 &\geq D_\varphi(B^{(t+1)},\bar{B}^{(t+1)})
\end{align*}
and hence we have
\begin{align*}
 D_\varphi(B^{(0)},\bar{B}^{(0)})  \geq  D_\varphi(B^{(1)},\bar{B}^{(1)}) 
 \geq  \cdots \geq D_\varphi(B^{(T)},\bar{B}^{(T)}). 
\end{align*}
This indicates that under a mild assumption 
the Bregman divergence $D_\varphi(B^{(t)},\bar{B}^{(t)})$ will converge to 
zero and that $B^{(t)}\in\mathcal{S}$ will also converge to a matrix in
$\mathcal{M}\cap\mathcal{S}$. A condition on the convergence has been investigated by
Bauschke, et al. \cite{bauschke02:_iterat_bregm_retrac}. 
This update algorithm is similar to the so-called
em-algorithm \cite{Amari95,dempster77:_maxim_likel_from_incom_data} which is a popular
algorithm in statistics and machine learning. In the em-algorithm, the $\eta$-projection
and the $\theta_V$-projection with $V(z)=-\log z$ is repeated in the probability space. 
Then, the maximum likelihood estimator under the partial observation is computed. 
In the context of statistical estimation, 
usually the em-algorithm is conducted when $\mathcal{M}\cap\mathcal{S}=\emptyset$ holds. 
Under some assumption with $\mathcal{M}\cap\mathcal{S}=\emptyset$, the point sequences 
$(B^{(t)},\bar{B}^{(t)})\in\mathcal{S}\times\mathcal{M}$ converges to the pair of 
the closest point $(B^{\ast},\bar{B}^{\ast})\in\mathcal{S}\times\mathcal{M}$ such that 
$(B^{\ast},\bar{B}^{\ast})$ is the optimal solution of the optimization problem,  
\begin{align*}
\min_{(B,\bar{B})\in\mathcal{S}\times\mathcal{M}}D_\varphi(B,\bar{B}), 
\end{align*}
see \cite{mclachlan08:_em_algor_and_exteny} for details. 
We believe that to provide a simple characterization about the convergence point
$(B^{\ast},\bar{B}^{\ast})$  under the condition
$\mathcal{M}\cap\mathcal{S}\neq\emptyset$ is an open problem. 

Next, we investigate Algorithm 2. Likewise we suppose $B_k=B^{(0)}\in\mathcal{S}$. Note
that $\mathcal{M}\cap\mathcal{S}$ is $\eta$-autoparallel. Let $B^{\star}$ be the
$\theta_\varphi$-projection of $B_k=B^{(0)}$ onto the intersection
$\mathcal{M}\cap\mathcal{S}$. Then the extended Pythagorean theorem leads that 
\begin{align*}
 D_\varphi(B^{\star},B^{(t)})
 &= 
 D_\varphi(B^{\star},\bar{B}^{(t)})+D_\varphi(\bar{B}^{(t)},B^{(t)})\\
 &=
 D_\varphi(B^{\star},B^{(t+1)})+D_\varphi(B^{(t+1)},\bar{B}^{(t)})
 +D_\varphi(\bar{B}^{(t)},B^{(t)})\\
 &\geq D_\varphi(B^{\star},B^{(t+1)})
\end{align*}
and hence we have
\begin{align*}
 D_\varphi(B^{\star},B^{(0)})\geq 
 D_\varphi(B^{\star},B^{(1)})\geq \cdots\geq
 D_\varphi(B^{\star},B^{(T)}). 
\end{align*}
Suppose that $B^{(T)}$ converges to $B^{(\infty)}\in\mathcal{M}\cap\mathcal{S}$ when $T$
tends to infinity, then the equality $B^{(\infty)}=B^{\star}$ holds as shown below. 
From the definition of $B^{\star}$ and the extended Pythagorean theorem, we have 
\begin{align*}
 D_{\varphi}(B^{(\infty)},B^{(T)})
 =D_{\varphi}(B^{(\infty)},B^*)+D_{\varphi}(B^*,B^{(T)}). 
\end{align*}
Due to the continuity of the Bregman divergence, for $T\rightarrow\infty$ we have 
\begin{align*}
 0=D_{\varphi}(B^{(\infty)},B^{(\infty)})
 =D_{\varphi}(B^{(\infty)},B^*)+D_{\varphi}(B^*,B^{(\infty)}), 
\end{align*}
and hence $B^{(\infty)}=B^{\star}$ holds. As the result we have
$\lim_{T\rightarrow\infty}B^{(T)}=B^{\star}$.  
Figure \ref{fig:geometry_sparse_quasi-Newton-boosting} shows the geometrical
illustration of the Algorithm 2. 
Applying Theorem 8.1 of Bauschke and Borwein \cite{bauschke97:_legen_funct_and_method_of},
we see that the convergence of $B^{(T)}$ to the point $B^{\star}$ is guaranteed under the
Bregman divergence associated with power potential with $\gamma\leq 0$. 
The iterative update procedure is closely related to the boosting algorithm
\cite{FreundSchapire97,murata04:_infor_geomet_u_boost_bregm_diver} in which the iterative
Bregman projection is exploited to compute the estimator for classification problems. 

\begin{figure}[t]
 \begin{center}
 \includegraphics[scale=0.5]{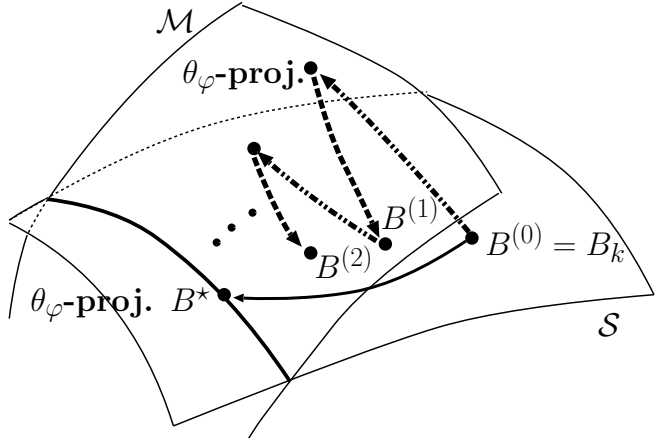}
 \end{center}
 \caption{Geometrical interpretation of Algorithm 2. 
 The sparse matrix $B^{(t)}$ will converge to $B^\star$
 which is the $\theta_\varphi$ projection of $B^{(0)}=B_k\in\mathcal{S}$.}
 \label{fig:geometry_sparse_quasi-Newton-boosting} 
\end{figure}

As argued above, it is not guaranteed that $B^{(t)}$ in Algorithm 1 converges to
$B^\star$, which is the $\theta_\varphi$-projection of $B_k=B^{(0)}$ onto
$\mathcal{M}\cap\mathcal{S}$. 
On the other hand the sequence $B^{(t)}$ in Algorithm 2 converges to $B^\star$ under mild
assumption. 
From the viewpoint of the least-change principle, 
the sparse quasi-Newton method with Algorithm 2 will be preferable. 
Fletcher \cite{fletcher95:_optim_posit_defin_updat_for} has proposed the sparse update
formula using $B^{\star}$. 
The update formula using the matrix $B^{\star}$ requires the sparsity and the secant
condition simultaneously, and hence, the approximate Hessian can be ill-posed when
$(s_k)_i=0$ for some $i$ \cite{sorensen82:_collin_scalin_and_sequen_estim}. 


\subsection{Computation of Projections}
\label{subsec:Computation_Projections}
We consider the computation of the extended sparse quasi-Newton updates. 
In Algorithm 1 and 2 above, we need to compute the $\theta_{\varphi}$-projection of a
matrix $B$ onto the $\eta$-autoparallel submanifold $\mathcal{S}$ consisting of sparse
positive definite matrices. 
Generally the $\theta_\varphi$-projection does not have the explicit expression. 
Here, we study only the $\theta_V$-projection based on the $V$-Bregman divergence. 

According to Yamashita \cite{yamashita08:_spars_quasi_newton_updat_with}, we briefly
introduce the computation of the projection onto $\mathcal{S}$, when the geometrical
structure is induced from the KL-divergence. 
For a given matrix $\bar{B}^{(t)}\in\mathcal{M}$, 
the projection onto $\mathcal{S}$, denoted as $B^{(t+1)}$, is obtained as the optimal solution of
\begin{align*}
 \min_{B\in\mathcal{\PD}(n)}\KL(B,\bar{B}^{(t)}),\qquad \st B\in\mathcal{S}. 
\end{align*}
Some calculation yields that $B^{(t+1)}$ is also the optimal solution of  
\begin{align*}
 \max_{B\in\mathcal{\PD}(n)}\det{B^{-1}},\qquad 
 \st (B^{-1})_{ij}=(H^{(t)})_{ij}\ \ (i,j)\in F. 
\end{align*}
Let $\bar{F}$ be $\bar{F}=F\backslash\{(i,i)~|~i=1,\ldots,n\}$. If the graph
$G=(\{1,\ldots,n\},\bar{F})$ is chordal, the existence of the optimal 
solution is guaranteed \cite{yamashita08:_spars_quasi_newton_updat_with,
fukuda00:_exploit_spars_in_semid_progr,
grone84:_posit_defin_compl_of_partial_hermit_matric}. 
The inverse of the optimal solution, $(B^{(t+1)})^{-1}$, is represented by using 
the sparse clique-factorization formula 
\cite{fukuda00:_exploit_spars_in_semid_progr,yamashita08:_spars_quasi_newton_updat_with},
and then the updated inverse Hessian matrix is obtained. 
The sparse clique-factorization formula of $(B^{(t+1)})^{-1}$ is represented by 
\begin{align*}
 (B^{(t+1)})^{-1}=
 L_1^\top L_2^\top \cdots L_{\ell-1}^\top  D L_{\ell-1} \cdots L_2 L_1
\end{align*}
in which $L_r\;(r=1,\ldots,\ell-1)$ are lower triangular matrices, and $D$ is 
a positive definite block-diagonal matrix consisting of $\ell$ diagonal blocks. 
The number of $\ell$ is determined by the the number of maximal cliques of the graph
$G=(\{1,\ldots,n\},\bar{F})$, and all elements of $L_r\;(r=1,\ldots,\ell-1)$ and $D$ are
explicitly computed from $(H^{(t)})_{ij},\,(i,j)\in F$. 
We generalize the above argument to the projection with the $V$-Bregman divergence. 
\begin{theorem}
 \label{theorem:min-det-B-chordal}
 Let $\bar{F}$ be $\bar{F}=F\backslash\{(i,i)~|~i=1,\ldots,n\}$, 
 and suppose that the undirected graph $(\{1,\ldots,n\},\bar{F})$ is chordal. 
 Let $\bar{B}^{(t)}\in\mathcal{M}$. Then there exists the $\theta_V$-projection 
 of $\bar{B}^{(t)}$ onto $\mathcal{S}$, and the projection is the optimal solution of the
 following problem, 
 \begin{align}
  \label{eqn:simplified-sparse-V-DFP}
  \begin{array}{l}
   \displaystyle
    \min_{B\in\mathrm{PD}(n)}\ \det(B),\quad
    \subto\  
    (\theta_V(B))_{ij}=(\theta_V(\bar{B}^{(t)}))_{ij},\  (i,j)\in F. 
  \end{array}
 \end{align}
\end{theorem}
\begin{proof}
Remember that $\theta_V(P)$ is defined as 
$\theta_V(P)=-\nu(\det{P})P^{-1}$ which is a negative definite matrix. 
It is easy to see that the mapping $-\theta_V(P)$ is bijection on $\PD(n)$. 
Hence, the assumption on the graph $(\{1,\ldots,n\},\bar{F})$ guarantees that 
 the problem
 \begin{align}
 \max_{B\in\PD(n)} \det(-\theta_V(B)),\qquad
 (\theta_V(B))_{ij}=(\theta_V(\bar{B}^{(t)}))_{ij}\;\; \text{for all}\;\; (i,j)\in F 
 \label{eqn:min-prob_det_theta_V}
 \end{align}
has the unique optimal solution $B^*$, 
and the optimal solution satisfies $(-\theta_V(B^*))^{-1}\in\mathcal{S}$, 
as shown in 
\cite{grone84:_posit_defin_compl_of_partial_hermit_matric,
fukuda00:_exploit_spars_in_semid_progr,yamashita08:_spars_quasi_newton_updat_with}. 
In terms of the objective function, we see that 
\begin{align*}
 \det(-\theta_V(B))= \det(\nu(\det{B})B^{-1})=\frac{\nu(\det{B})^n}{\det{B}}. 
\end{align*}
 The function $\nu(z)^n/z$ is strictly monotone decreasing for $z>0$. Indeed, 
\begin{align*}
 \frac{d}{dz}\log\frac{\nu(z)^n}{z}=\frac{n}{z}\bigg(\beta(z)-\frac{1}{n}\bigg)<0
\end{align*}
holds. Thus, the optimal solution of \eqref{eqn:min-prob_det_theta_V} is identical to that
of \eqref{eqn:simplified-sparse-V-DFP}. 
We find that $B^*\in\mathcal{S}$ holds, since
 $(-\theta_V(B^*))^{-1}=\nu(\det{B^*})^{-1}B^*\in\mathcal{S}$ holds. 
For any $B\in\mathcal{S}$, we have 
\begin{align*}
D_V(B,\bar{B}^{(t)})-D_V(B,B^*)- D_V(B^*,\bar{B}^{(t)})
 &=\sum_{i,j}(\theta_V(\bar{B}^{(t)})-\theta_V(B^*))_{ij}(B^*-B)_{ij}\\
 &=\sum_{(i,j)\not\in F}(\theta_V(\bar{B}^{(t)})-\theta_V(B^*))_{ij}(B^*-B)_{ij}\\
 &=0. 
\end{align*}
The second and third equalities follows 
$(\theta_V(\bar{B}^{(t)})-\theta_V(B^*))_{ij}=0$ for $(i,j)\in F$ and
$(B^*-B)_{ij}=0$ for $(i,j)\not\in F$, respectively. 
Therefore, $B^*$ is identical to the $\theta_V$-projection of $\bar{B}_t$ onto
$\mathcal{S}$. 
\end{proof}

We present a practical method of computing the projection of $\bar{B}^{(t)}$ onto
$\mathcal{S}$. Let $B^{(t)}$ and $\bar{B}^{(t)}$ for $t=0,1,2,\ldots$ be 
matrices generated by the extended sparse quasi-Newton update with Algorithm 2. 
We show a method of computing $H^{(t)}=(B^{(t)})^{-1}$ and
$\bar{H}^{(t)}=(\bar{B}^{(t)})^{-1}$. 
Suppose we have $H^{(t)}$, then $\bar{H}^{(t)}$ is obtained by solving the problem
\begin{align*}
 \min_{H\in\PD(n)} D_V(H^{-1},(H^{(t)})^{-1}),\qquad Hy_k=s_k. 
\end{align*}
In the similar way of the proof of Theorem \ref{theorem:V-BFGS-form}, the optimal solution
$\bar{H}^{(t)}$ satisfies 
\begin{align*}
\bar{H}^{(t)}=\frac{\nu(\det({\bar{H}^{(t)}})^{-1})}{\nu(\det({H^{(t)}})^{-1})}
 B^{\mathrm{DFP}}[H^{(t)};y_k,s_k]+
 \bigg(1-\frac{\nu(\det({\bar{H}^{(t)}})^{-1})}{\nu(\det({H^{(t)}})^{-1})}\bigg)
 \frac{s_ks_k^\top}{s_k^\top y_k}. 
\end{align*}
We need only the elements $(\bar{H}^{(t)})_{ij}$ for $(i,j)\in F$ and the determinant
$\det(\bar{H}^{(t)})$. 
If we have the Choleskey factorization or the sparse clique-factorization formula of
$H^{(t)}$, we can obtain these values by simple computation. 
Then, the matrix $H^{(t+1)}$ is given as the optimal solution of 
\begin{align*}
 \min_{H\in\PD(n)} D_V(H^{-1},(\bar{H}^{(t)})^{-1}),\qquad H^{-1}\in\mathcal{S}. 
\end{align*}
As shown in the proof of Theorem \ref{theorem:min-det-B-chordal}, 
$H^{(t+1)}$ is also the optimal solution of 
\begin{align*}
 \max_{H\in\PD(n)} \det(-\theta_V(H^{-1})),\qquad
 \theta_V(H^{-1})_{ij}=\theta_V((\bar{H}^{(t)})^{-1})_{ij}
 \;\; \text{for all}\;\; (i,j)\in F 
\end{align*}
Let $X=-\theta_V((H^{(t+1)})^{-1})=\nu(\det(H^{(t+1)})^{-1})H^{(t+1)}$, then 
the sparse clique-factorization formula provides the factorized expression of $X$ 
based on the information of 
$\nu(\det(\bar{H}^{(t)})^{-1})\bar{H}^{(t)}_{ij},\,(i,j)\in F$. 
The determinant of $X$ is easily computed by the sparse clique-factorization formula. 
Then, we solve the the following equation, 
\begin{align*}
 \det{X}=\frac{\nu(z)^n}{z},\quad z>0. 
\end{align*}
The Newton method is available to find the unique solution $z^*$ efficiently. 
Using the solution $z^*$, the matrix $H^{(t+1)}$ is represented 
\begin{align*}
 H^{(t+1)}=\frac{1}{\nu(z^*)}X. 
\end{align*}
The matrix $H^{(t+1)}$ also has the expression of the sparse clique-factorization
formula, and thus, it is available to the sequel computation. 
%

\section{Concluding Remarks}
\label{sec:Concluding_Remarks}
Along the line of the research stared by Fletcher
\cite{fletcher91:_new_resul_for_quasi_newton_formul}, we considered the quasi-Newton
update formula based on the Bregman divergences, and presented a geometrical
interpretation of the Hessian update formulae. We studied the invariance property of the
update formulae. The sparse quasi-Newton methods were also considered based on the
information geometry. We show that the information geometry is useful tool 
not only to better understand the quasi-Newton methods but also 
to design new update formulae. 

As pointed out in Section \ref{sec:potential-function_quasi-Newton}, 
the self-scaling quasi-Newton method with the popular scaling parameter 
is out of the formulae derived from the Bregman divergence. 
Nocedal and Yuan proved that the self-scaling quasi-Newton method with the popular scaling
parameter has some drawbacks \cite{nocedal93:_analy_of_self_scalin_quasi_newton_method}. 
An interesting future work is to pursue the relation between the numerical properties and
the geometrical structure behind the optimization algorithms. 
In the study of the interior point methods, it has been made clear that geometrical
viewpoint is useful \cite{ohara10:_infor_geomet_approac_to_polyn}. 
The geometrical viewpoint will become important to investigate algorithms for numerical
computation.

\section{Acknowledgements}
The authors are grateful to Dr.~Nobuo Yamashita of Kyoto university for helpful comments. 
T.~Kanamori was partially supported by Grant-in-Aid for Young Scientists
(20700251).

\appendix 
\section{Proof of Theorems \ref{theorem:V-BFGS-form}}
\label{appendix:VBFGS-formula}
We prove the following lemma which is useful to show the existence of the optimal
solution. 
\begin{lemma}
 \label{lemma:sol_existence-nu-equation}
 Let $V$ be a potential and $\nu=\nu_V$. 
 For any $C>0$ the equation
 \begin{align}
  \label{eqn:lemma-equation-any-C}
 C\nu(z)^{n-1}=z,\quad z>0
 \end{align}
has the unique solution. 
\end{lemma}
\begin{proof}
 We define the function $\zeta(z)$ by $\zeta(z)=\log z-(n-1)\log \nu(z)$, 
 then, the \eqref{eqn:lemma-equation-any-C} is equivalent to the equation
 \begin{align}
  \label{eqn:expression2-equation-any-C}
  \log C=\zeta(z),\quad z>0. 
 \end{align}
 Since the potential function satisfies $\lim_{z\rightarrow+0}z/\nu(z)^{n-1}=0$ from the
 definition, we have $\lim_{z\rightarrow+0}\zeta(z)=-\infty$. 
 In terms of the derivative of $\zeta(z)$, we have the following inequality
 \begin{align*}
  \frac{d}{dz}\zeta(z)=\frac{1}{z}-(n-1)\frac{\beta(z)}{z}>\frac{1}{zn}>0. 
 \end{align*}
 Thus, $\zeta(z)$ is an increasing function on $\Real_+$. 
 Moreover we have
\begin{align*}
 \zeta(z)\geq \zeta(1)+\int_1^z\frac{1}{zn}dz=
 \zeta(1)+\frac{\log z}{n}. 
\end{align*}
 The above inequality implies that $\lim_{z\rightarrow\infty}\zeta(z)=\infty$. 
Since $\zeta(z)$ is continuous, the equation \eqref{eqn:expression2-equation-any-C} has
the unique solution. 
\end{proof}

\begin{proof}
[Proof of Theorem \ref{theorem:V-BFGS-form}]
First, we show the existence of the matrix $B_{k+1}$ 
satisfying \eqref{eqn:update-formula-V-BFGS}. Lemma \ref{lemma:sol_existence-nu-equation} now shows 
that there exists a solution $z^*>0$ for the equation 
\begin{align*}
 \frac{\det(B^{BFGS}[B_k;s_k,y_k])}{\nu(\det{B_k})^{n-1}}\cdot\nu(z)^{n-1}=z,\quad z>0. 
\end{align*}
By using the solution $z^*$, we define the matrix $\bar{B}$ such that
\begin{align*}
 \bar{B}=\frac{\nu(z^*)}{\nu(\det{B_k})}B^{BFGS}[B_k;s_k,y_k]
 +\big(1-\frac{\nu(z^*)}{\nu(\det{B_k})}\big)\frac{y_ky_k^\top}{s_k^\top y_k}, 
\end{align*}
then the determinant of $\bar{B}$ satisfies
\begin{align*}
 \det{\bar{B}}=\frac{\det(B^{BFGS}[B_k])}{\nu(\det{B_k})^{n-1}}\cdot \nu(z^*)^{n-1}=z^*, 
\end{align*}
in which the first equality comes from the formula 
$\det(A+vu^\top)= \det(A)(1+u^\top A^{-1}v)$ 
and the second one follows the definition of $z^*$. 
Hence there exists $B_{k+1}\in\mathrm{PD}(n)$ satisfying \eqref{eqn:update-formula-V-BFGS}. 

Next, we show that the matrix $B_{k+1}$ in \eqref{eqn:update-formula-V-BFGS} 
satisfies the optimality condition of \eqref{eqn:V-BFGS-hessian-update-prob}. 
According to G\"{u}ler, et al. \cite{guler09:_dualit_in_quasi_newton_method},
the normal vector for the affine subspace 
\[
\mathcal{M}=\{B\in\mathrm{PD}(n)~|~Bs_k=y_k\}
\]
is characterized by the form of 
\begin{align}
 \label{eqn:normal-vector-secant-cond}
 s_k\lambda^\top+\lambda s_k^\top\in\mathrm{Sym}(n),\qquad \lambda\in\Real^n. 
\end{align}
Suppose $B'\in\mathrm{PD}(n)$ be an optimal solution of 
\eqref{eqn:V-BFGS-hessian-update-prob}, then $B'$ satisfies the optimality condition that
 there exists a vector $\lambda\in\Real^n$ such that 
\begin{align*}
 &\phantom{\Longleftrightarrow}
 \nabla_B D_V(B,B_k)\big|_{B=B'}=s_k\lambda^\top + \lambda s_k^\top
 \nonumber\\
 &\Longleftrightarrow\ 
 -\nu(\det(B'))(B')^{-1}+\nu(\det(B_k))B_k^{-1} 
 =s_k\lambda^\top + \lambda s_k^\top, 
\end{align*}
where $\nabla_B D_V(B,B_k)$ denotes the gradient of $D_V(B,B_k)$ with respect to the
variable $B$. Also, the optimal solution $B'$ should satisfy the constraint $B's_k=y_k$. 
On the other hand, the matrix $B_{k+1}$ defined by
 \eqref{eqn:update-formula-V-BFGS} satisfies
\begin{align*}
 B_{k+1}^{-1} 
 &=
 \frac{\nu(\det{B_k})}{\nu(\det{B_{k+1}})}(B^{BFGS}[B_k;s_k,y_k])^{-1}
 +\bigg(1-\frac{\nu(\det{B_k})}{\nu(\det{B_{k+1}})}\bigg)
 \frac{s_ks_k^\top}{s_k^\top y_k}\\
&=
 \frac{\nu(\det{B_k})}{\nu(\det{B_{k+1}})}B^{DFP}[B_k^{-1};y_k,s_k]
 +\bigg(1-\frac{\nu(\det{B_k})}{\nu(\det{B_{k+1}})}\bigg)
 \frac{s_ks_k^\top}{s_k^\top y_k}\\
\Longleftrightarrow &
 \left\{
 \begin{array}{l}
  \displaystyle  
   -\nu(\det{B_{k+1}})B_{k+1}^{-1}+
   \nu(\det{B_k})B_k^{-1}
   =s_k\lambda^\top+\lambda s_k^\top,\vspace*{2mm}\\
  \displaystyle  
   \lambda
   =\frac{\nu(\det{B_k})}{s_k^\top y_k}B_k^{-1}y_k
   -\frac{\nu(\det{B_{k+1}})}{2s_k^\top y_k}s_k
   -\frac{\nu(\det{B_k})y_k^\top B_k^{-1}y_k}{2(s_k^\top y_k)^2}s_k. 
 \end{array}
 \right.
\end{align*}
The conditions $s_k^\top y_k>0$ and $B_k\in\mathrm{PD}(n)$ guarantees 
the existence of the above vector $\lambda$. In addition, the direct computation yields
that the constraint $B_{k+1}s_k=y_k$ is satisfied. 
Hence, $B_{k+1}$ satisfies the optimality condition. 
Since \eqref{eqn:V-BFGS-hessian-update-prob} is a strictly convex problem, 
$B_{k+1}$ is the unique optimal solution. 
\end{proof}

\bibliographystyle{plain}

\end{document}